\documentclass[journal,12pt,onecolumn,draftclsnofoot]{IEEEtran}
\usepackage[cmex10]{amsmath,mathtools} 
\makeatletter
\newcommand{\svast}{\bBigg@{3}}
\newcommand{\vast}{\bBigg@{4}}
\newcommand{\Vast}{\bBigg@{5}}

\makeatother

\usepackage{amsthm, amssymb, amsfonts} 
\usepackage{bm} 
\usepackage[ruled]{algorithm2e} 
\usepackage{cite} 
\usepackage[pdftex]{graphicx} 
\usepackage{epstopdf} 
\usepackage{color}
\usepackage{xcolor}
\usepackage[keeplastbox]{flushend}
\usepackage{enumitem} 
\usepackage[caption=false]{subfig} 

\usepackage{hyperref} 
\hypersetup{colorlinks=true, linkcolor=black, citecolor=black,anchorcolor=black,urlcolor=black}  

\newtheorem{theorem}{Theorem}
\newtheorem{corollary}{Corollary}[theorem]
\usepackage{cleveref} 
\crefname{equation}{eq.}{eqs.}
\crefname{figure}{fig.}{figs.}

\newcommand{\RNum}[1]{\MakeUppercase{\romannumeral#1}} 
\begin{document}

\title{Physical-Layer Security for Two-Hop Air-to-Underwater Communication Systems With Fixed-Gain Amplify-and-Forward Relaying}

\author{
    Yi~Lou,~\IEEEmembership{Member,~IEEE,}
    Ruofan~Sun,~\IEEEmembership{Student Member,~IEEE,}
    Julian~Cheng,~\IEEEmembership{Senior~Member,~IEEE,}
    Songzuo~Liu,~\IEEEmembership{Member,~IEEE,}
    Feng~Zhou,~\IEEEmembership{Member,~IEEE,}
    and~Gang Qiao,~\IEEEmembership{Member,~IEEE}
    \thanks{Yi Lou, Ruofan Sun, Songzuo Liu, Feng Zhou,  and Gang Qiao are with the Acoustic Science and Technology Laboratory, Harbin Engineering University, Harbin 150001, China, and also with the Key Laboratory of Marine Information Acquisition and Security (Harbin Engineering University), Ministry of Industry and Information Technology, Harbin 150001, China (e-mail: \{louyi,sunruofan,liusongzuo,zhoufeng,qiaogang\}@hrbeu.edu.cn).}
    \thanks{Julian Cheng is with the School of Engineering, The University of British Columbia, Kelowna, BC, Canada (e-mail: julian.cheng@ubc.ca).}
}

\maketitle

\begin{abstract}
We analyze a secure two-hop mixed radio frequency (RF) and underwater wireless optical communication (UWOC) system using a fixed-gain amplify-and-forward (AF) relay. The UWOC channel is modeled using a unified mixture exponential-generalized Gamma distribution to consider the combined effects of air bubbles and temperature gradients on transmission characteristics. Both legitimate and eavesdropping RF channels are modeled using flexible $\alpha-\mu$ distributions. Specifically, we first derive both the probability density function (PDF) and cumulative distribution function (CDF) of the received signal-to-noise ratio (SNR) of the mixed RF and UWOC system. Based on the PDF and CDF expressions, we derive the closed-form expressions for the tight lower bound of the secrecy outage probability (SOP) and the probability of non-zero secrecy capacity (PNZ), which are both expressed in terms bivariate Fox's $H$-function. To utilize these analytical expressions, we derive asymptotic expressions of SOP and PNZ using only elementary functions. Also, we use asymptotic expressions to determine the optimal transmitting power to maximize energy efficiency. Further, we thoroughly investigate the effect of levels of air bubbles and temperature gradients in the UWOC channel, and study nonlinear characteristics of the transmission medium and the number of multipath clusters of the RF channel on the secrecy performance. Finally, all analyses are validated using Monte Carlo simulation.
\end{abstract}


\IEEEpeerreviewmaketitle

\section{Introduction}
\IEEEPARstart{T}{he} rise of the underwater Internet of Things requires the support of a high-performance underwater communication network with high data rates, low latency, and long communication range. Underwater wireless optical communication (UWOC) is one of the essential technologies for this communication network. Unlike radio frequency (RF) and acoustic technologies, UWOC technology can achieve ultra-high data rates of Gpbs over a moderate communication range when selecting blue or green light with wavelengths located in the transmission window \cite{zengSurveyUnderwaterOptical2017}. Further, a light-emitting diode or laser diode as a light source provides the versatility to switch between maximizing the communication range or the coverage area within the constraints of the range-beamwidth tradeoff to meet the needs of a specific application scenario.

Using relay technology to construct a communication system in a multi-hop fashion is one of the primary techniques to extend the communication range. Based on the modality of processing and forwarding signals, relays can be divided into two main categories: decode-and-forward relays (DF) and amplify-and-forward (AF) relays. In DF relaying systems, the relay down-converts the received signals to the baseband, decodes, re-encodes, and up-converts them to the RF band, which are then forwarded to the destination node. In AF relaying systems, the relay amplifies the received signals directly in the passband based on the amplification factor, then forwards them directly in the RF band. Since AF scheme does not require time-consuming decoding and spectral shifting, it can significantly reduce complexity while still providing good performance \cite{dohlerCooperativeCommunicationsHardware2010}. Depending on the different CSI information required by the AF relay, AF relaying can be divided into the variable gain AF (VG) one and fixed-gain AF (FG) one. In a VG scheme, the relay requires instantaneous  channel state information (CSI) of the source-to-relay link, whereas in an FG scheme, only statistical CSI of the SR link is required \cite{shinMRCAnalysisCooperative2008}. Therefore, from an engineering standpoint, the FG scheme is more attractive because of its low implementation complexity.

To maximize the utilization of the different transmission environments of each hop and thus improve the overall performance of the multi-hop relaying system, mixed communication systems using different communication technologies have been proposed. For example, the mixed communication system using both RF and free space optical (FSO) technologies is proposed to take advantage of the robustness of the RF links and the high bandwidth characteristics of the FSO links. Further, RF sub-systems offer low-cost and non-line-of-sight communication capabilities, while FSO sub-systems offer low transmission latency and ultra-high transmission rates. Therefore, a mixed RF/FSO system is a cost-effective solution to the last-mile problem in wireless communication networks, where the high-bandwidth FSO sub-system of a mixed RF/FSO system is used to connect seamlessly the fiber backbone and RF sub-system access networks \cite{upadhyaEffectInterferenceMisalignment2020,leiSecrecyOutageAnalysis2018a,zediniPerformanceAnalysisDualHop2016,djordjevicMixedRFFSO2015}. In the past, achieving ultra-high-speed communication between underwater and airborne nodes or land-based stations across the sea surface medium has been a challenge due to the low data rate of underwater acoustic communications. To solve this problem, using an ocean buoy or a marine ship as a relay node,
the mixed RF/UWOC system is proposed, in which the high-speed UWOC is used instead of underwater acoustic communication, to achieve higher overall communication rates \cite{leiPerformanceAnalysisDualHop2020,illiPhysicalLayerSecurity2020,illiDualHopMixedRFUOW2018,christopoulouOutageProbabilityMultisensor2019,xingAdaptiveEnergyEfficientAlgorithm2018}. 

Accurate modeling of the UWOC channel, including absorption, scattering, and turbulence, is a prerequisite for proper performance analysis and algorithm development of the UWOC system \cite{l.johnsonSurveyChannelModels21,zediniNewSimpleModel04}. Absorption and scattering have been extensively studied \cite{jaruwatanadilokUnderwaterWirelessOptical2008,cochenourTemporalResponseUnderwater2013,
nabaviEmpiricalModelingAnalysis2019}, where absorption limits the transmission distance of underwater light, while scattering diffuses the receiving radius of underwater light transmission and deflects the transmission path, thus reducing the received optical power. Due to changes in the random refractive index variation, turbulence can cause fluctuations in the received irradiance, i.e., scintillation, which can limit the performance and affect the stability of the UWOC system \cite{zengSurveyUnderwaterOptical2017}. In early research, UWOC turbulence was modeled by borrowing  models of atmospheric turbulence, e.g., weak turbulence is modeled by the Lognormal distribution \cite{jamaliPerformanceAnalysisMultiHop2017,jamaliPerformanceStudiesUnderwater2017,nezamalhosseiniOptimalPowerAllocation2020,jiangPerformanceSpatialDiversity2020}, and moderate-to-strong turbulence is modeled by the Gamma-Gamma distribution \cite{luanScintillationIndexOptical2019,boucouvalasUnderwaterOpticalWireless2016,shinStatisticalModelingImpact2020,elamassieVerticalUnderwaterVisible2020}.

However, the statistical distributions used to model atmospheric turbulence cannot accurately characterize UWOC systems due to the fundamental differences between aqueous and atmospheric mediums. Recently, based on experimental data, the mixed exponential-lognormal distribution has been proposed to model moderate to strong UWOC turbulence in the presence of air bubbles in both fresh water and salty water \cite{jamaliStatisticalStudiesFading2018}. Later,  the mixture exponential-generalized Gamma (EGG) distribution was proposed to model turbulence in the presence of air bubbles and temperature gradients in either fresh or salt water \cite{zediniUnifiedStatisticalChannel2019}. The EGG distribution not only can model turbulence of various intensities, but also has an analytically tractable mathematical form. Therefore, useful system performance metrics, such as ergodic capacity, outage probability, and BER, can be easily obtained. 

Due to the broadcast nature of RF signals, secrecy performance has always been one of the most important considerations for the mixed RF/FSO communication systems \cite{hamamrehClassificationsApplicationsPhysical2019,leiSecrecyPerformanceMixed2017,leiSecrecyOutageAnalysis2018,yangPhysicalLayerSecurityMixed2018,leiSecrecyOutageAnalysis2018a,leiSecureMixedRFFSO2020}. In \cite{leiSecrecyPerformanceMixed2017}, the expressions of the lower bound of the secrecy outage probability (SOP) and average secrecy capacity (ASC) for mixed RF/FSO systems using VG or FG relaying schemes, were both derived in closed-form, where the RF and FSO links are modeled by the Nakagami-$m$ and GG distributions, respectively. The authors in \cite{leiSecrecyOutageAnalysis2018} used Rayleigh and GG distributions to model RF and FSO links, respectively. Considering the impact of imperfect channel state information (CSI), both the exact and asymptotic expressions of the lower bound for SOP of a mixed RF/FSO system using VG or FG relay are derived. The same authors then extended the analysis to multiple-input and multiple-output configuration and analyzed the impact of different transmit antenna selection schemes on the secrecy performance of the mixed RF/FSO system using a DF relay, where RF and FSO links are modeled by the Nakagami-$m$ and $\mathcal{M}$-distributions, respectively. Assuming the CSI of the FSO and RF links are imprecise and outdated, the authors derived the bound and asymptotic expressions of the effective secrecy throughput of the system. In \cite{yangPhysicalLayerSecurityMixed2018}, using more generalized $\eta$-$\mu$ and $\mathcal{M}$-distributions to model RF and FSO links, respectively, and assuming that the eavesdropper is only at the relay location, the authors derived the analytical results for the SOP and the average secrecy rate of the mixed RF/FSO system using the FG or VG relaying scheme.  To quantify the impact of the energy harvesting operation on the system secrecy performance, the authors in \cite{leiSecrecyOutageAnalysis2018a} derived exact closed-form and asymptotic expressions for the SOP of the downlink simultaneous wireless information and power transfer system using DF relaying scheme, under the assumption that RF and FSO links are modeled using the Nakagami-$m$ and GG distributions, respectively.

However, research on the secrecy performance of mixed RF/UWOC systems is still in its infancy despite the growing number of underwater communication applications. The authors in \cite{illiDualHopMixedRFUOW2018} investigated the secrecy performance of a two-hop mixed RF/UWOC system using a VG or FG multiple-antennas relay and maximal ratio combining scheme, where RF and UWOC links are modeled by Nakagami-$m$ and the mixed exponential-Gamma (EG) distributions, respectively. Assuming that only the source-to-relay link receives eavesdropping from unauthorized users, the authors in \cite{illiDualHopMixedRFUOW2018} derived the exact closed expressions of the ASC and SOP of the mixed RF/UWOC systems. Later, based on the same channel model as in \cite{illiDualHopMixedRFUOW2018}, the same authors extended the  analysis to the mixed RF/UWOC system using a multi-antennas DF relay with the selection combining scheme \cite{illiPhysicalLayerSecurity2020}. Both the exact closed-form and asymptotic expressions of the SOP were derived.

However, while the EG distribution is suitable for modeling turbulence of various intensities in both fresh water and salty water, this distribution fails to model the effects of air bubbles and temperature gradients on UWOC turbulence \cite{zediniUnifiedStatisticalChannel2019}. Further, the Nakagami-$m$ distribution is only applicable to certain specific scenarios and cannot accurately characterize the effects of the properties of the transmission medium and multipath clusters on channel fading. It is shown that the impact of the medium on the signal propagation is mainly determined by the nonlinearity characteristics of the medium \cite{yacoubAmDistributionPhysical2007}. The $\alpha-\mu$ distribution is a more general, flexible, and mathematically tractable  model of channel fading whose parameters $\alpha$ and $\mu$ are correlated with the nonlinearity of the propagation medium and the number of clusters of multipath transmission, respectively. Further, by setting $\alpha$ and $\mu$ to specific values, the $\alpha-\mu$ distribution can be reduced to several classical channel fading models, including Nakagami-$m$, Gamma, one-sided Gaussian, Rayleigh, and Weibull distributions. Recently, the secrecy performance of a two-hop mixed RF/UWOC system using DF relaying has been analyzed in \cite{louSecrecyOutageAnalysis2020}; however, only the lower bound and asymptotic expressions of the SOP are derived. Furthermore, the overall end-to-end latency of the mixed RF/UWOC communication system is increased by the decoding and forwarding and spectral shifting operations required by DF relaying.

However, to the best of the authors' knowledge, this is the first comprehensive secrecy performance analysis of the mixed RF/UWOC communications system using a low-complexity FG relaying scheme. Unlike previous UWOC channel models that do not adequately characterize the underwater optical propagation and RF channel models that use various simplifying assumptions, we model the RF channels and the UWOC channel using the more general and accurate $\alpha-\mu$ and EGG distributions, respectively, to analyze the effects of a variety of real channel physics phenomena, such as different temperature gradients and levels of air bubbles of UWOC channels and different grades of medium nonlinearity, and the number of multipath clusters of the RF channels on the secrecy performance of the mixed RF/UWOC communication systems. We propose a novel analytical framework to derive the closed-form expressions of the SOP and the non-zero secrecy capacity (PNZ) metrics by the bivariate Fox’s $H$-function. Moreover, our secrecy performance study provides a generalized framework for several fading models for both RF and UWOC channels, such as Rayleigh, Weibull for RF channels and EG and Generalized Gamma for UWOC channels. We first derive the probability density function (PDF) and cumulative distribution function (CDF) of the end-to-end SNR for the mixed RF/UWOC communication system in exact closed-form in terms of bivariate $H$-function. Depending on these expressions, we derive the exact closed-form expressions of the lower bound of the SOP and the PNZ. Furthermore, we also derive asymptotic expressions for both SOP and PNZ containing only simple functions at high SNRs. Also, based on the asymptotic expressions for SOP and PNZ, we provide a straightforward approach to determine the optimal source transmission power to maximize energy efficiency for given performance goals of both SOP and PNZ. Finally, we use Monte Carlo simulation to validate all the derived analytical expressions and theoretical analyses.

The rest of this paper is organized as follows. In Section \RNum{2}, the channel and system models are presented. In Section \RNum{3}, the end-to-end statistics are studied. Both exact and asymptotic expressions for the SOP and PNZ are derived in Section \RNum{4}. The numerical results and discussions are discussed in Section \RNum{5}, which is followed by the conclusion in Section \RNum{6}.

\section{System and Channel Models}
A mixed RF/UWOC system is considered in Fig. 1 where the source node (S) in the air transmits its private data to the legitimate destination node (D) located underwater via a trusted relay node (R), which can be a buoy or a surface ship. The RF channel from S to R and underwater optical channel from the R to the D node is assumed to follow $\alpha-\mu$ and EGG distributions, respectively. During transmission, one unauthorized receiver (E) attempts to eavesdrop on RF signals received by the R. In this paper, we consider a VG AF relay where the relay amplifies the received signal by a fixed factor and then forwards the amplified message to the destination node.

\begin{figure}[!t]
    \includegraphics[width=.45\textwidth]{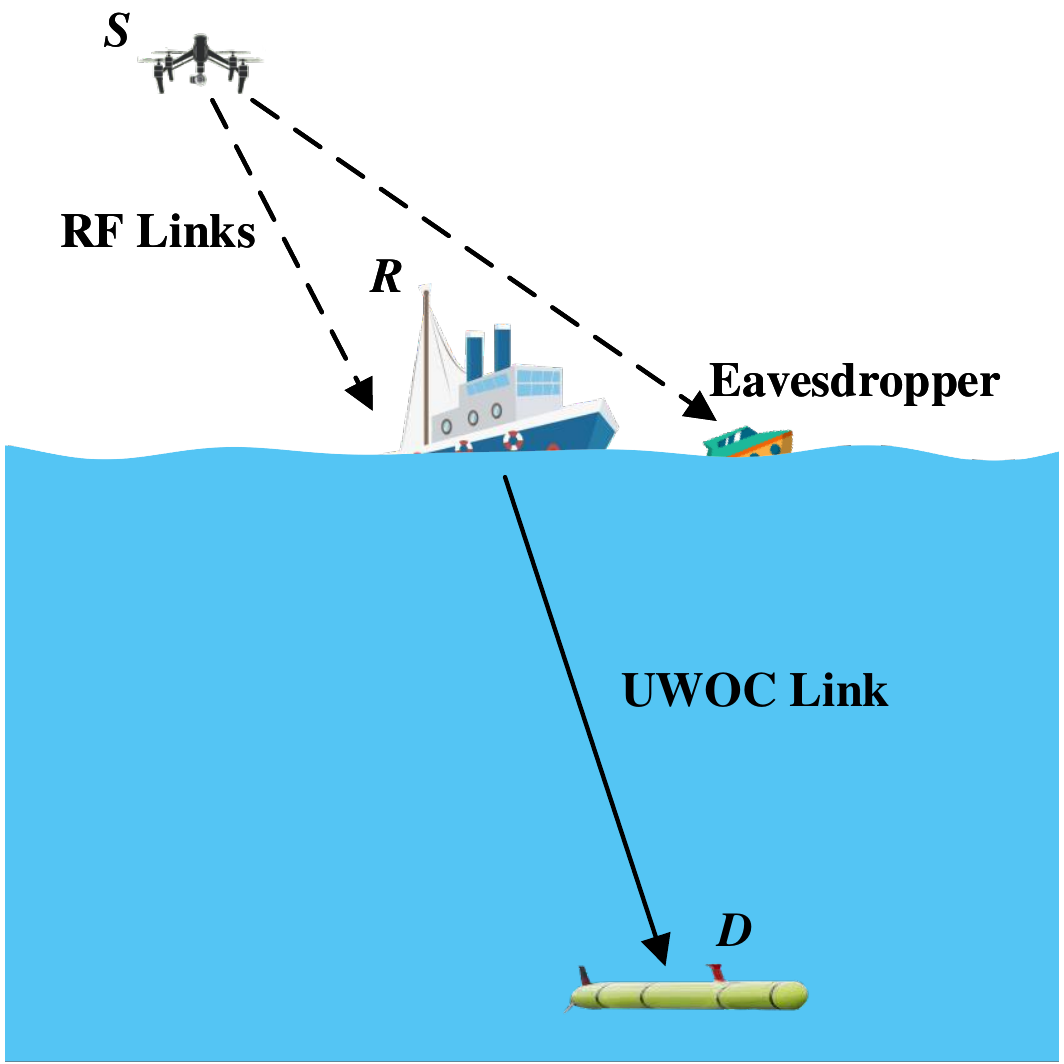}
    \centering
    \caption{A mixed RF/UWOC two-hop communication system using an FG relaying scheme with one legitimate receiver in the presence of eavesdropping.}
    \label{fig: 1}
\end{figure}

\subsection{RF Channel Model}
The RF SR link is modeled by $\alpha$-$\mu$ flat fading models, where the PDF of the received SNR, denoted by $\gamma_1$, can be expressed as
\begin{IEEEeqnarray*}{rcl}
    f_{\gamma _1}\left(\gamma _1\right)&=&\frac{\alpha }{2\Gamma (\mu )}\frac{\mu ^{\mu }}{\left(\bar{\gamma }_1\right)^{\frac{\alpha \mu }{2}}}\gamma _1^{\frac{\alpha \mu }{2}-1}\thinspace\text{exp}\left(-\mu \left(\frac{\gamma _1}{\bar{\gamma }_1}\right)^{\frac{\alpha }{2}}\right)\IEEEyesnumber \label{ampdf}
\end{IEEEeqnarray*}
where $\gamma_1 \geq 0$, $\mu \geq 0$, $\alpha \geq 0$, and $\Gamma(\cdot)$ denotes the gamma function. The fading model parameters $\alpha$ and $\mu$ are associated with the non-linearity and multi-path propagation of the channel. Further, the PDF of the received SNR at the eavesdropping node E, denoted by $f_{\gamma _e}(\gamma _e)$, also follows $\alpha$-$\mu$ with parameters $\alpha_e$ and $\mu_e$.

Based on the definition of the $H$-function, the CDF of $\gamma_1$, which is defined as $F_{\gamma_1}(\gamma_1)=\int_{0}^{\gamma_1} f_{\gamma_1}(\gamma_1) d \gamma_1$, can be expressed as
\begin{IEEEeqnarray*}{rcl}
	 F_{\gamma _1}\left(\gamma _1\right)&\overset{(a)}{=}&\kappa \int _0^{\gamma }  H_{0,1}^{1,0}\left[\gamma  \Lambda \left|\begin{array}{c}  \\ \left(-\frac{1}{\alpha }+\mu ,\frac{1}{\alpha }\right) \\\end{array}\right.\right]d\gamma  \\
	 &=&-\frac{i \kappa }{2 \pi }\int _{\mathcal{L}}^s\!\!\Lambda ^{-s} \Gamma \left(\frac{s}{\alpha }+\mu -\frac{1}{\alpha }\right) \!\!\int _0^{\gamma }\!\!\gamma^{-s}d\gamma ds \\
	 &=&\frac{i \kappa }{2 \pi }\int _{\mathcal{L}}^s\frac{\gamma ^{1-s} \Lambda ^{-s} }{s-1}\Gamma \left(\frac{s}{\alpha }+\mu -\frac{1}{\alpha }\right)ds\\
	 &=&\frac{\kappa  }{\Lambda }H_{1,2}^{1,1}\left[\gamma  \Lambda \left|\begin{array}{c} (1,1) \\ \left(\mu ,\frac{1}{\alpha }\right),(0,1) \\\end{array}\right.\right] \IEEEyesnumber \label{amcdfform2}
\end{IEEEeqnarray*}
where we use \cite[Eq. (1.60)]{mathaiHFunctionTheoryApplications2010} and \cite[Eq. (1.125)]{mathaiHFunctionTheoryApplications2010} to express $f_{\gamma_1}(\gamma_1)$ in the right side of equity (a) into the form of $H$-function, where $H_{\cdot ,\cdot }^{\cdot ,\cdot }[\cdot|\cdot]$ is the $H$-Function \cite[Eq. (1.2)]{mathaiHFunctionTheoryApplications2010}, $\kappa = \frac{\beta }{\Gamma (\mu )\bar{\gamma }}$, $ \Lambda = \frac{\beta }{\bar{\gamma }}$, and $\beta = \frac{\Gamma \left(\frac{1}{\alpha }+\mu \right)}{\Gamma (\mu )} $. 
Note that, the present form of $F_{\gamma _1}\left(\gamma _1\right)$ in \eqref{amcdfform2} is more suitable for deriving secrecy performance of a two-hop mixed RF/UWOC than the form proposed in \cite[Eq. (2)]{kongHighlyAccurateAsymptotic2018} for the point-to-point system over single-input multiple-output $\alpha-\mu$ channels.

\subsection{UWOC channel model}
To characterize the combined effects of different levels of air bubbles and temperature gradients on the light intensity received at underwater node D, we model the UWOC channel from R to D using the EGG distribution \cite{zediniUnifiedStatisticalChannel2019}, where the PDF of the received SNR, defined as $\gamma_2=\left(\eta I\right)^{r} / N_{0}$, has been derived in closed-form in terms of Meijer-$G$ functions \cite[Eq. (3)]{zediniPerformanceAnalysisDualHop2020}. Based on \cite[Eq. (1.112)]{mathaiHFunctionTheoryApplications2010}, we can re-write the PDF of $\gamma_2$ using $H$-functions as
\begin{IEEEeqnarray*}{rcl}
    f_{\gamma _2}(\gamma _2)&=&\frac{c (1-\omega )  }{\gamma  r \Gamma (a)}H_{0,1}^{1,0}\!\!\left[b^{-c}\left(\frac{\gamma _2}{\mu _r}\right)^{\frac{c}{r}}\middle|\!\!\!\begin{array}{c}   \\ (a,1) \\\end{array}\!\!\!\right]\\
    &&+\frac{\omega   }{\gamma_2  r}H_{0,1}^{1,0}\!\!\left[\frac{1}{\lambda }\left(\frac{\gamma _2}{\mu _r}\right)^{\frac{1}{r}}\middle|\!\!\!\begin{array}{c}   \\ (1,1) \\\end{array}\!\!\!\right] \IEEEyesnumber
\end{IEEEeqnarray*}
where the parameters $\omega$, a, b and c can be estimated using the maximum likelihood criterion with expectation maximization algorithm. The parameter $\omega$ is the mixed weight of the distribution; $\lambda$ is the parameter related to the exponential distribution; parameters $a$, $b$, and $c$ are related to the exponential distribution; $r$ is a parameter dependent on the detection scheme, specifically, $r=1$ for heterodyne detection and $r=2$ for intensity modulation and direct detection \cite[Eq. (31)]{lapidothCapacityFreeSpaceOptical2009a}.

The EGG distribution can provide the best fit with the measured data form laboratory water tank experiments in the presence of temperature gradients and air bubbles \cite{zediniPerformanceAnalysisDualHop2020}. Therefore, by using the EGG distribution to model the UWOC link, we can gain more insight into the relationship between characteristics of the UWOC link and the secrecy performance of the mixed RF/UWOC communication system.

Using the definition of CCDF, i.e., $\bar{F}_{\gamma_2}(\gamma_1)=\int_{0}^{\gamma_1} f_{\gamma_1}(\gamma_1) d \gamma_1$, and an approach similar to that used to derive \eqref{amcdfform2}, we can derive the CCDF of $\gamma_2$ as 
\begin{IEEEeqnarray*}{rcl}
	 \bar{F}_{\gamma _2}\left(\gamma _2\right)&=&-\frac{i (1-\omega )}{2 \pi  \Gamma (a)}\int _{\mathcal{L}}^s\!\!\Gamma \left(a+\frac{r s}{c}\right) b^{rs}\mu _r^s \int _{\gamma }^{\infty }\!\!\gamma ^{-s-1}d\gamma ds\\
	 &&-\frac{i \omega }{2 \pi }\int _{\mathcal{L}}^s\Gamma (r s+1) \lambda ^{rs}\mu _r^s \int _{\gamma }^{\infty }\gamma ^{-s-1}d\gamma ds \\
	 &=&-\frac{i (1-\omega ) }{2 \pi  \Gamma (a)}\int _{\mathcal{L}}^s\frac{\gamma ^{-s} }{s}\Gamma \left(a+\frac{r s}{c}\right) \left(b^r \mu _r\right)^sds\\
	 &&-\frac{i\omega  }{2 \pi }\int _{\mathcal{L}}^s\frac{\gamma ^{-s}}{s}\Gamma (r s+1) \left(\lambda ^r \mu _r\right)^sds \\
	 &=&\frac{(1-\omega ) }{\Gamma (a)}H_{1,2}^{2,0}\!\!\left[\frac{b^{-r} \gamma }{\mu _r}\middle|\!\!\!\begin{array}{c} (1,1) \\ (0,1),(a,\frac{r}{c}) \\\end{array}\!\!\!\right]\\
	 &&+ r \omega  H_{0,1}^{1,0}\!\!\left[\frac{\gamma  \lambda ^{-r}}{\mu _r}\middle|\!\!\!\begin{array}{c}   \\ (0,r) \\\end{array}\!\!\!\right]. \IEEEyesnumber \label{eggccdf}
\end{IEEEeqnarray*}
It is worth to mention that the newly derived expression in \eqref{eggccdf} is useful to derive the closed-form CDF expression of the end-to-end SNR of the mixed RF/UWOC communication system.

\section{End-To-End SNR}
In this section, we derive the exact closed-form expressions for PDF and CDF of the end-to-end SNR of  mixed RF/UWOC communication system. We then use these expressions to derive closed-form and asymptotic expressions for the system secrecy metrics in the following section.

The end-to-end instantaneous SNR of the mixed RF/UWOC system using the FG relaying scheme is given as
\begin{IEEEeqnarray*}{rcl}
    \gamma_{e q}=\frac{\gamma_{1} \gamma_{2}}{\gamma_{2}+C}\IEEEyesnumber \label{gammaeq1}
\end{IEEEeqnarray*}
where $C$ denotes the FG amplifying constant and is inversely proportional to the square of the relay transmitting power, and this constant is defined as $C=1 /\left(G^{2} N_{0}\right)$, where the FG amplifying factor $G$ is defined as
\begin{IEEEeqnarray*}{rcl}
	 G=\sqrt{\frac{P_{2}}{P_{1} |h_1|^{2}+N_{1}}}.\IEEEyesnumber
\end{IEEEeqnarray*}

Using the definition of $H$-function, we can express $G$ in terms of the $H$-functions
\begin{IEEEeqnarray*}{rcl}
	G=\kappa  H_{1,2}^{2,1}\left[\Lambda \left|\begin{array}{c} (0,1) \\ (0,1),(-\frac{1}{\alpha }+\mu ,\frac{1}{\alpha }) \\\end{array}\right.\right].\IEEEyesnumber \label{G}
\end{IEEEeqnarray*}

\begin{theorem}
  The CDF of the end-to-end SNR of the mixed RF/UWOC communcation system using the FG relaying scheme $F_{\gamma_{e q}}(\gamma_{e q})$, deﬁned in \eqref{gammaeq1}, can be obtained in exact closed-form as shown in \eqref{gammaeqcdf} in terms of bivariate $H$-functions, where $H_{\cdot,\cdot:\cdot,\cdot;\cdot,\cdot}^{\cdot,\cdot:\cdot,\cdot;\cdot,\cdot}[\cdot|\cdot]$ is the bivariate $H$-Function defined as \cite[Eq. (2.55)]{mathaiHFunctionTheoryApplications2010}
\end{theorem}

\begin{proof}
See Appendix A. 
\end{proof}

Note that the current implementation of bivariate $H$-function for numerical computation is mature and efficient, including GPU-accelerated versions, and has been implemented using the most popular software, including MATLAB$^\circledR$  \cite{cherguiRicianKfactorbasedAnalysis2018}, Mathematica$^\circledR$ \cite{almeidagarciaCACFARDetectionPerformance2019}, and Python \cite{alhennawiClosedFormExactAsymptotic2016}. Also, the exact-closed expression for the CDF in \eqref{gammaeqcdf} is a key analytical tool to derive the SOP metric of the mixed RF/UWOC system.

\begin{figure*}[!bht]
    \begin{IEEEeqnarray*}{rcl}
	F_{\gamma _{eq}}(\gamma _{eq})&=&1-\frac{\gamma  \kappa  (1-\omega )  }{\Gamma (a)}H_{1,0:0,2;1,1}^{0,1:2,0;0,1}\!\!\left[\begin{array}{c} \frac{b^{-r} C}{\mu _r} \\ \frac{1}{\gamma  \Lambda } \\\end{array}\middle|\!\!\!\begin{array}{ccccc} (2,1,1) & : &   & ; & \left(1+\frac{1}{\alpha }-\mu ,\frac{1}{\alpha }\right) \\
	   & : & (0,1),(a,\frac{r}{c}) & ; & (1,1) \\\end{array}\!\!\!\right]\\
	   &&-\gamma  \kappa  r \omega  H_{1,0:0,2;1,1}^{0,1:2,0;0,1}\!\!\left[\begin{array}{c} \frac{C \lambda ^{-r}}{\mu _r} \\ \frac{1}{\gamma  \Lambda } \\\end{array}\middle|\!\!\!\begin{array}{ccccc} (2,1,1) & : &   & ; & \left(1+\frac{1}{\alpha }-\mu ,\frac{1}{\alpha }\right) \\
	   & : & (1,1),(0,r) & ; & (1,1) \\\end{array}\!\!\!\right].\IEEEyesnumber \label{gammaeqcdf}
	\end{IEEEeqnarray*}
    \hrule
    \begin{IEEEeqnarray*}{rcl}
	f_{\gamma _{eq}}(\gamma _{eq})&=&\frac{\kappa  (1-\omega )  }{\Gamma (a)}H_{1,0:1,1;0,2}^{0,1:0,1;2,0}\!\!\left[\begin{array}{c} \frac{1}{\gamma  \Lambda } \\ \frac{b^{-r} C}{\mu _r} \\\end{array}\middle|\!\!\!\begin{array}{ccccc} (2,1,1) & : & \left(1+\frac{1}{\alpha }-\mu ,\frac{1}{\alpha }\right) & ; &   \\
	   & : & (2,1) & ; & (0,1),(a,\frac{r}{c}) \\\end{array}\!\!\!\right]\\
	   &&+\kappa  r \omega  H_{1,0:1,1;0,2}^{0,1:0,1;2,0}\!\!\left[\begin{array}{c} \frac{1}{\gamma  \Lambda } \\ \frac{C \lambda ^{-r}}{\mu _r} \\\end{array}\middle|\!\!\!\begin{array}{ccccc} (2,1,1) & : & \left(1+\frac{1}{\alpha }-\mu ,\frac{1}{\alpha }\right) & ; &   \\
	   & : & (2,1) & ; & (1,1),(0,r) \\\end{array}\!\!\!\right].\IEEEyesnumber\label{gammaeqpdf}
	\end{IEEEeqnarray*}
	\hrule
	\begin{IEEEeqnarray*}{rcl}
	 \text{SOP}_L&=&1-\frac{\Theta  \kappa  (1-\omega ) \kappa _e }{\Gamma (a) \Lambda _e^2}H_{1,0:1,2;0,2}^{0,1:1,1;2,0}\!\!\left[\begin{array}{c} \frac{\Lambda _e}{\Theta  \Lambda } \\ \frac{b^{-r} C}{\mu _r} \\\end{array}\middle|\!\!\!\begin{array}{ccccc} (2,1,1) & : & \left(1+\frac{1}{\alpha }-\mu ,\frac{1}{\alpha }\right) & ; &   \\
	   & : & \left(\frac{1+\alpha _e \mu _e}{\alpha _e},\frac{1}{\alpha _e}\right),(1,1) & ; & (0,1),(a,\frac{r}{c}) \\\end{array}\!\!\!\right] \\
	 &&-\frac{\Theta  \kappa  r \omega  \kappa _e }{\Lambda _e^2}H_{1,0:1,2;0,2}^{0,1:1,1;2,0}\!\!\left[\begin{array}{c} \frac{\Lambda _e}{\Theta  \Lambda } \\ \frac{C \lambda ^{-r}}{\mu _r} \\\end{array}\middle|\!\!\!\begin{array}{ccccc} (2,1,1) & : & \left(1+\frac{1}{\alpha }-\mu ,\frac{1}{\alpha }\right) & ; &   \\
	   & : & \left(\frac{1+\alpha _e \mu _e}{\alpha _e},\frac{1}{\alpha _e}\right),(1,1) & ; & (1,1),(0,r) \\\end{array}\!\!\!\right]. \IEEEyesnumber \label{soplb}
\end{IEEEeqnarray*}
\hrule
\begin{IEEEeqnarray*}{rcl}
	\text{SOP}_a=1-\frac{\kappa  (1-\omega ) \kappa _e }{\Lambda  \Gamma (a) \Lambda _e}H_{3,2}^{1,3}\!\!\left[\frac{b^r \Lambda _e \mu _r}{C \Theta \Lambda }\middle|\!\!\!\begin{array}{c} (1,1),(1-a,\frac{r}{c}),(1-\mu ,\frac{1}{\alpha }) \\ \left(\mu _e,\frac{1}{\alpha _e}\right),(0,1) \\\end{array}\!\!\!\right]-\frac{\kappa  r \omega  \kappa _e }{\Lambda  \Lambda _e}H_{2,1}^{1,2}\!\!\left[\frac{\lambda ^r \Lambda _e \mu _r}{C \Theta  \Lambda }\middle|\!\!\!\begin{array}{c} (1,r),(1-\mu ,\frac{1}{\alpha }) \\ \left(\mu _e,\frac{1}{\alpha _e}\right) \\\end{array}\!\!\!\right].\IEEEyesnumber\label{sopa}
\end{IEEEeqnarray*}
\hrule
\begin{IEEEeqnarray*}{rcl}
	\text{SOP}_{\text{ae}}&=&1-\frac{(1-\omega ) \alpha _e}{\Gamma (a) \Gamma (\mu ) \Gamma \left(\mu _e\right) \Gamma \left(\alpha _e \mu _e+1\right)}\Gamma\left(\mu +\frac{\alpha _e \mu _e}{\alpha }\right) \left(\frac{\Lambda _e}{\Theta  \Lambda }\right)^{\alpha _e \mu _e} H_{2,1}^{1,2}\!\!\left[\frac{b^r\mu _r}{C}\middle|\!\!\!\begin{array}{c} (1,1),(1-a,\frac{r}{c}) \\ \left(\alpha _e \mu _e,1\right) \\\end{array}\!\!\!\right]\\
	&&-\frac{r \omega  \alpha _e}{\Gamma (\mu ) \Gamma \left(\mu _e\right) \Gamma \left(\alpha _e \mu _e+1\right)}\Gamma \left(\mu +\frac{\alpha_e \mu _e}{\alpha }\right) \left(\frac{\Lambda _e}{\Theta  \Lambda }\right)^{\alpha _e \mu _e} H_{2,1}^{1,2}\!\!\left[\frac{\lambda ^r \mu _r}{C}\middle|\!\!\!\begin{array}{c} (0,1),(1,r) \\ \left(\alpha _e \mu _e,1\right) \\\end{array}\!\!\!\right].\IEEEyesnumber\label{sopae}
\end{IEEEeqnarray*}
\hrule
\end{figure*}

\begin{theorem}
  The PDF of the end-to-end SNR, which is defined in \eqref{gammaeq1}, of the mixed RF/UWOC communication system using the FG relaying scheme, denoted by $f_{\gamma_{e q}}(\gamma_{e q})$, can be obtained in exact closed-form as shown in \eqref{gammaeqpdf}.
\end{theorem}

\begin{proof}
See Appendix B. 
\end{proof}

It is worth noting that the PDF expression in \eqref{gammaeqpdf} is the most critical step required to evaluate the PNZ performance metric, as will be shown in the next section.

\section{PERFORMANCE METRICS}
This section presents analytical expressions for the critical secrecy performance metrics of a mixed RF/UWOC communication system, including both SOP and PNZ, in the presence of air bubbles and temperature gradients in the UWOC channel and medium nonlinearity in the RF channel.

\subsection{SOP}
SOP is defined as the probability that the secrecy capacity $C_s$ falls below a target rate of conﬁdential information $R_s$ and it can be expressed as
\begin{IEEEeqnarray*}{rcl}
    S O P\left(R_{s}\right)&=&\operatorname{Pr}\left\{\log _{2}\left(\frac{1+\gamma_{eq}}{1+\gamma_{e}}\right)<R_{s}\right\} \\ &=&\operatorname{Pr}\left\{\gamma_{eq} \leq \Theta \gamma_{e}+\Theta-1\right\} \\
    &=&\int_{0}^{\infty} \!\!\!F_{eq}\left(\Theta \gamma_{e}+\Theta-1\right) f_{e}\left(\gamma_{e}\right) d \gamma_{e}\IEEEyesnumber
\end{IEEEeqnarray*}
where $\Theta=e^{R_{s}}$.

\subsubsection{Lower bound}
Referring to \cite{leiSecrecyCapacityAnalysis2017,kongPhysicalLayerSecurity2019}, a tight lower bound for the SOP can be given as derived as
\begin{IEEEeqnarray*}{rcl}
    S O P_{L}=\int_{0}^{\infty} F_{\gamma_{e q}}(\Theta \gamma) f_{\gamma_{e}}(\gamma) d \gamma.\IEEEyesnumber\label{soplb1}
\end{IEEEeqnarray*}

\begin{theorem}
  The lower bound for the SOP of the mixed RF/UWOC communication  system using the FG relaying scheme deﬁned in \eqref{soplb1} can be obtained in exact closed-form as shown in \eqref{soplb}.
\end{theorem}

\begin{proof}
See Appendix C. 
\end{proof}

\subsubsection{Asymptotic results}
To gain more insight into the SOP performance and the dependency between the link quality of both RF and UWOC channels, we now derive asymptotic expressions for SOP. We consider two scenarios, namely $\gamma_1 \to \infty$ and $\gamma_e \to \infty$.
\begin{corollary}
For scenarios $\gamma_1 \to \infty$ and $\gamma_e \to \infty$, the asymptotic expressions of SOP of a mixed RF/UWOC communication system using FG relaying scheme can be given as \eqref{sopa} and \eqref{sopae} in terms of $H$-functions, respectively.
\end{corollary}
\begin{proof}
See Appendix D. 
\end{proof}

Note that in contrast to the closed expression of the lower bound of the SOP in \eqref{soplb} in terms of bivariate $H$-functions, which requires numerical evaluation of double line integrals, the asymptotic expressions in \eqref{sopa} and \eqref{sopae} only require numerical calculation of single line integrals, thus reducing the complexity of the calculations. Furthermore, as shown in Section \RNum{5}, for a target SOP performance, the asymptotic expressions in \eqref{sopa} and \eqref{sopae} can be used to determine rapidly the optimal transmitting power to maximize energy efficiency.

In the special case of a two-hop mixed RF/UWOC communication system over Rayleigh RF links and a thermally uniform UWOC channel, we can further simplify the asymptotic expressions in \eqref{sopa} and \eqref{sopae} by setting $c=1$, $\alpha=\alpha_e=1$, $\mu=\mu_e=1$. For example, eq. \eqref{sopa} can be simplified into
\begin{IEEEeqnarray*}{rcl}
	\text{SOP}_a&=&\frac{\kappa  \kappa _e}{\lambda  \Lambda  \Lambda _e^2 \mu _r}\svast(\!-a \lambda  (1-\omega ) \Lambda _e \mu _r \thinspace\text{exp}\left(\frac{C \Theta  \Lambda}{b \Lambda _e \mu _r}\right)\\
	&&\times E_{a+1}\left(\frac{C \Theta  \Lambda }{b \Lambda _e \mu _r}\right)-C \Theta  \Lambda  \omega  \thinspace\text{exp}\left(\frac{C \Theta  \Lambda}{\lambda  \Lambda _e \mu _r}\right) \\
	&&\times E_i\left(-\frac{C \Theta  \Lambda }{\lambda  \Lambda _e \mu _r}\right)\svast)+\frac{\Lambda  \Lambda _e-\kappa \omega  \kappa _e}{\Lambda  \Lambda _e}\IEEEyesnumber
\end{IEEEeqnarray*}
where $E_i(x)$ and $E_n(x)$ both denote the exponential integral \cite[Eq. (8.211.1)]{i.s.gradshteynTableIntegralsSeries2007}

\subsection{PNZ}
PNZ is another critical  metric for to evaluate the secrecy performance of a communication system, which is defined as $\operatorname{Pr}\left(C_{s}>0\right)$, where $C_s$ is the secrecy capacity. PNZ is generally related to channel conditions of all the channels in the mixed RF/UWOC systems. In this section, we derive the exact closed-form and asymptotic expressions for PNZ and analyze the relationship between channel parameters and PNZ performance.

\subsubsection{Exact results}
According to \cite{kongPhysicalLayerSecurity2019}, PNZ can be reformed as
\begin{IEEEeqnarray*}{rcl}
	\mathcal{P}_{n z}=\operatorname{Pr}\left(\gamma_{eq}>\gamma_{e}\right)=\int_{0}^{\infty} f_{eq}\left(\gamma_{eq}\right) F_{e}\left(\gamma_{eq}\right) d \gamma_{eq}. \IEEEyesnumber \label{pnz1}
\end{IEEEeqnarray*}
\begin{theorem}
  The exact PNZ of the mixed RF/UWOC communication system using the FG relaying scheme deﬁned in \eqref{pnz1} can be obtained in exact closed-form as shown in \eqref{pnz}.
\end{theorem}

\begin{proof}
See Appendix E. 
\end{proof}

\subsubsection{Asymptotic results}
To gain more insight into the PNZ performance and the dependency between the link quality of both RF and UWOC channels, we now derive asymptotic expressions for PNZ. We consider two scenarios, namely $\gamma_1 \to \infty$ and $\gamma_e \to \infty$.
\begin{corollary}
For scenarios $\gamma_1 \to \infty$ and $\gamma_e \to \infty$, the asymptotic expressions of PNZ of a mixed RF/UWOC communication system using the FG relaying scheme are given as \eqref{pnza} and \eqref{pnzae} in terms of $H$-functions, respectively.
\end{corollary}
\begin{proof}
Observing that the expressions for the lower bound of the SOP in \eqref{soplb} and exact PNZ in \eqref{pnz} have a similar structure; therefore, eqs. \eqref{pnza} and \eqref{pnzae} can be easily obtained using the same techniques as those used for deriving \eqref{sopa} and \eqref{sopae}, and the proof is complete.
\end{proof}

Note that, similar to the asymptotic expressions of the SOP in \eqref{sopa} and \eqref{sopae}, for a target PNZ performance, the asymptotic expressions of PNZ in \eqref{pnza} and \eqref{pnzae} are also suitable for fast numerical calculations and useful to determine the optimal transmitting power to maximize energy efficiency.

\begin{figure*}[!bht]
\begin{IEEEeqnarray*}{rcl}
	\mathcal{P}_{n z}&=&\frac{\kappa  (1-\omega ) \kappa _e }{\Gamma (a) \Lambda _e^2}H_{1,0:0,2;1,2}^{0,1:2,0;1,1}\!\!\left[\begin{array}{c} \frac{b^{-r} C}{\mu _r} \\ \frac{\Lambda _e}{\Lambda } \\\end{array}\middle|\!\!\!\begin{array}{ccccc} (2,1,1) & : &   & ; & \left(1+\frac{1}{\alpha }-\mu ,\frac{1}{\alpha }\right) \\
	   & : & (0,1),(a,\frac{r}{c}) & ; & \left(\frac{1+\alpha _e \mu _e}{\alpha _e},\frac{1}{\alpha _e}\right),(1,1) \\\end{array}\!\!\!\right]\\
	   &&+\frac{\kappa  r \omega  \kappa _e }{\Lambda _e^2}H_{1,0:0,2;1,2}^{0,1:2,0;1,1}\!\!\left[\begin{array}{c} \frac{C \lambda ^{-r}}{\mu _r} \\ \frac{\Lambda _e}{\Lambda } \\\end{array}\middle|\!\!\!\begin{array}{ccccc} (2,1,1) & : &   & ; & \left(1+\frac{1}{\alpha }-\mu ,\frac{1}{\alpha }\right) \\
	   & : & (1,1),(0,r) & ; & \left(\frac{1+\alpha _e \mu _e}{\alpha _e},\frac{1}{\alpha _e}\right),(1,1) \\\end{array}\!\!\!\right].\IEEEyesnumber \label{pnz}
\end{IEEEeqnarray*}
\hrule
\begin{IEEEeqnarray*}{rcl}
	\mathcal{P}_{n z a}&=&\frac{\kappa  r \omega  \kappa _e }{\Lambda  \Lambda _e}H_{2,1}^{1,2}\!\!\left[\frac{\lambda ^r \Lambda _e \mu _r}{C \Lambda }\middle|\!\!\!\begin{array}{c} (1,r),(1-\mu ,\frac{1}{\alpha }) \\ \left(\mu _e,\frac{1}{\alpha _e}\right) \\\end{array}\!\!\!\right]+\frac{\kappa  (1-\omega ) \kappa _e }{\Lambda  \Gamma (a) \Lambda _e}H_{3,2}^{1,3}\!\!\left[\frac{b^r \Lambda _e \mu _r}{C \Lambda }\middle|\!\!\!\begin{array}{c} (1,1),(1-a,\frac{r}{c}),(1-\mu ,\frac{1}{\alpha }) \\ \left(\mu _e,\frac{1}{\alpha _e}\right),(0,1) \\\end{array}\!\!\!\right].\IEEEyesnumber \label{pnza}
\end{IEEEeqnarray*}
\hrule
\begin{IEEEeqnarray*}{rcl}
	\mathcal{P}_{n z a e}&=&\frac{\kappa  r \omega  \alpha _e \kappa _e \Lambda ^{-\alpha _e \mu _e-1} \Lambda _e^{\alpha _e \mu _e-1} \Gamma \left(\mu +\frac{\alpha _e \mu_e}{\alpha }\right) }{\Gamma \left(\alpha _e \mu _e+1\right)}H_{2,1}^{1,2}\!\!\left[\frac{\lambda ^r \mu _r}{C}\middle|\!\!\!\begin{array}{c} (0,1),(1,r) \\ \left(\alpha _e \mu _e,1\right) \\\end{array}\!\!\!\right]\\
	&&+\frac{\kappa  (1-\omega ) \alpha _e \kappa _e \Lambda ^{-\alpha _e \mu _e-1} \Lambda _e^{\alpha _e \mu _e-1} \Gamma \left(\mu +\frac{\alpha_e \mu _e}{\alpha }\right) }{\Gamma (a) \Gamma \left(\alpha _e \mu _e+1\right)}H_{2,1}^{1,2}\!\!\left[\frac{b^r \mu _r}{C}\middle|\!\!\!\begin{array}{c} (1,1),(1-a,\frac{r}{c}) \\ \left(\alpha _e \mu _e,1\right) \\\end{array}\!\!\!\right].\IEEEyesnumber \label{pnzae}
\end{IEEEeqnarray*}
\hrule
\end{figure*}

\section{Numerical Results And Discussion}
 In this section, we provide some numerical results to verify the analytic and asymptotic expressions of SOP and PNZ derived in Section \RNum{4}, and thoroughly investigate the combined effect of the channel quality of both RF and UWOC channels on the secrecy performance of the two-hop mixed RF/UWOC communication system. All practical environmental physical factors that can affect channel quality, including levels of air bubbles, temperature gradients, and salinity of the UWOC channel, as well as the medium nonlinearity and multipath cluster characteristics of the RF channel, are taken into account. For brevity, we use $[\cdot, \cdot]$ to denote the value set of $[\text{air bubbles level}, \text{temperature gradient}]$ in this section.

In Fig. \ref{fig.2} -- Fig. \ref{fig.6}, we investigate the combined effect of the channel quality of both RF and UWOC channels on the SOP metric of the  two-hop mixed RF/UWOC communication system.

\begin{figure}[!htp]
\centering
\includegraphics[width=3.45in]{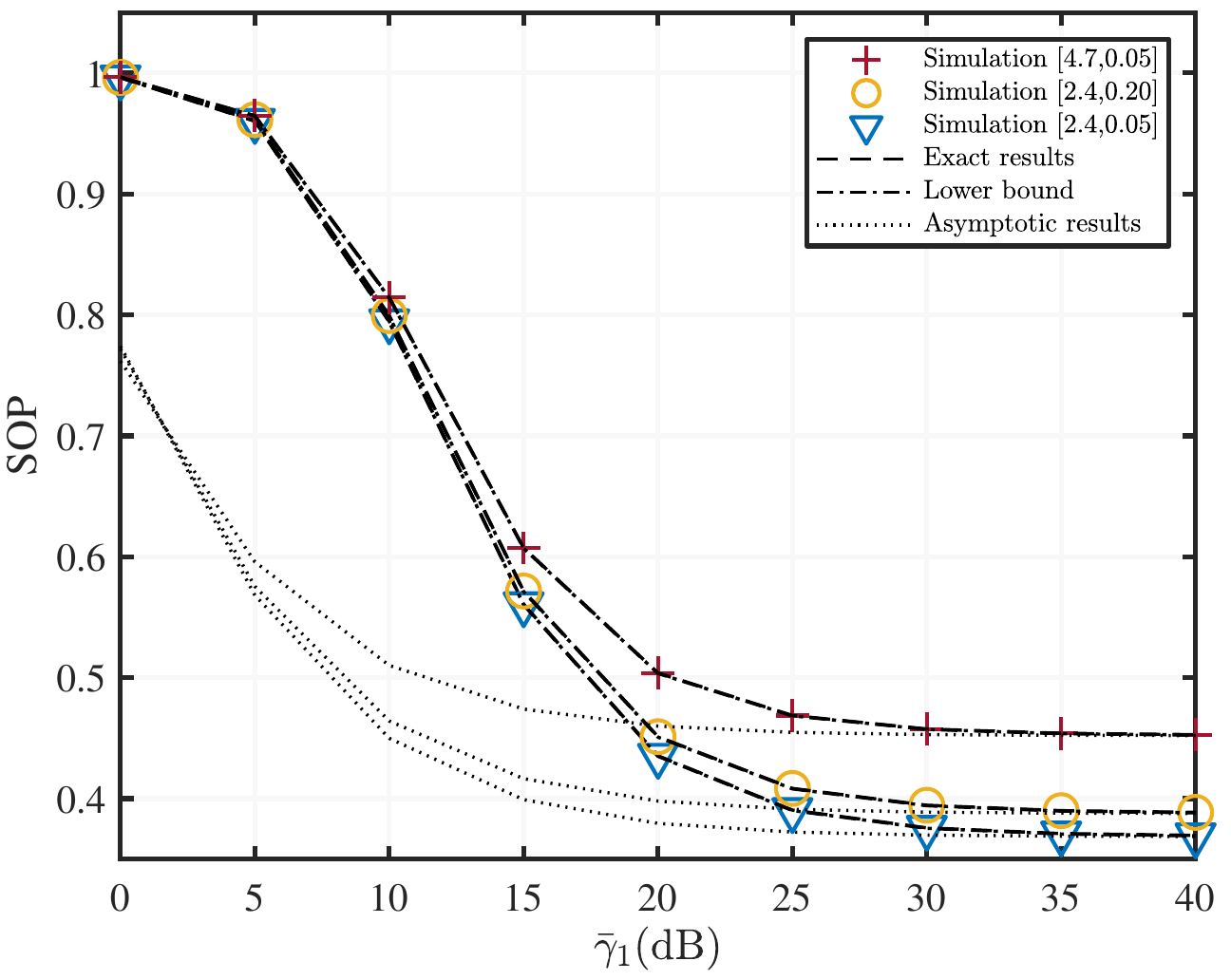}
\caption{SOP versus $\bar{\gamma}_1$ with various fading parameters when $\alpha=\alpha_e=1.6$, $\mu=\mu_e=1.5$, $R_s=0.01$, and $\bar{\gamma}_e=\bar{\gamma}_2
=10$ dB}
\label{fig.2}
\end{figure}

\begin{figure}[!htp]
\centering
\includegraphics[width=3.45in]{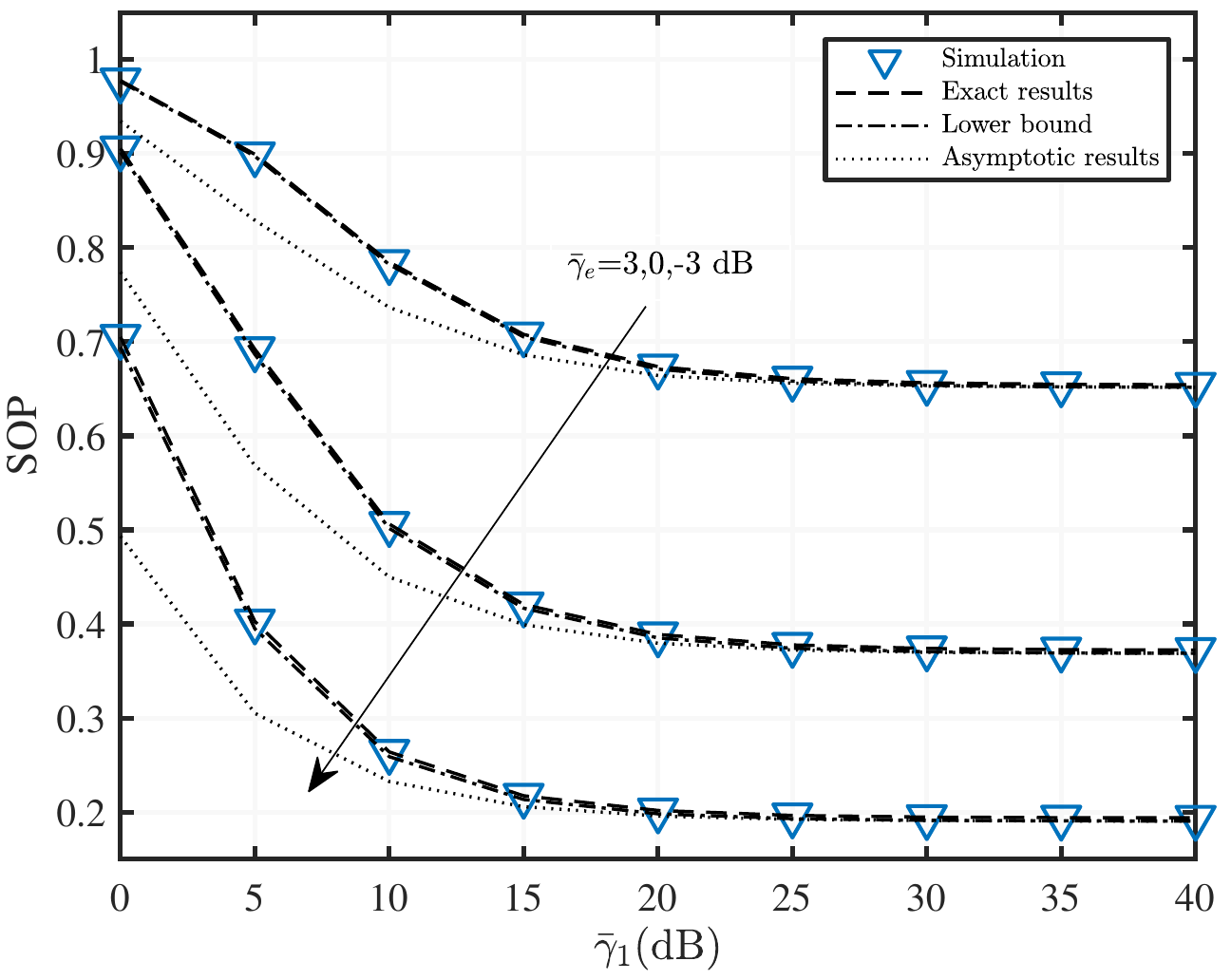}
\caption{SOP versus $\bar{\gamma}_1$ with various fading parameters when $\alpha=\alpha_e=1.6$, $\mu=\mu_e=1.5$, $R_s=0.01$, and $\bar{\gamma}_2=0$ dB}
\label{fig.3}
\end{figure}

\begin{figure}[!htp]
\centering
\includegraphics[width=3.45in]{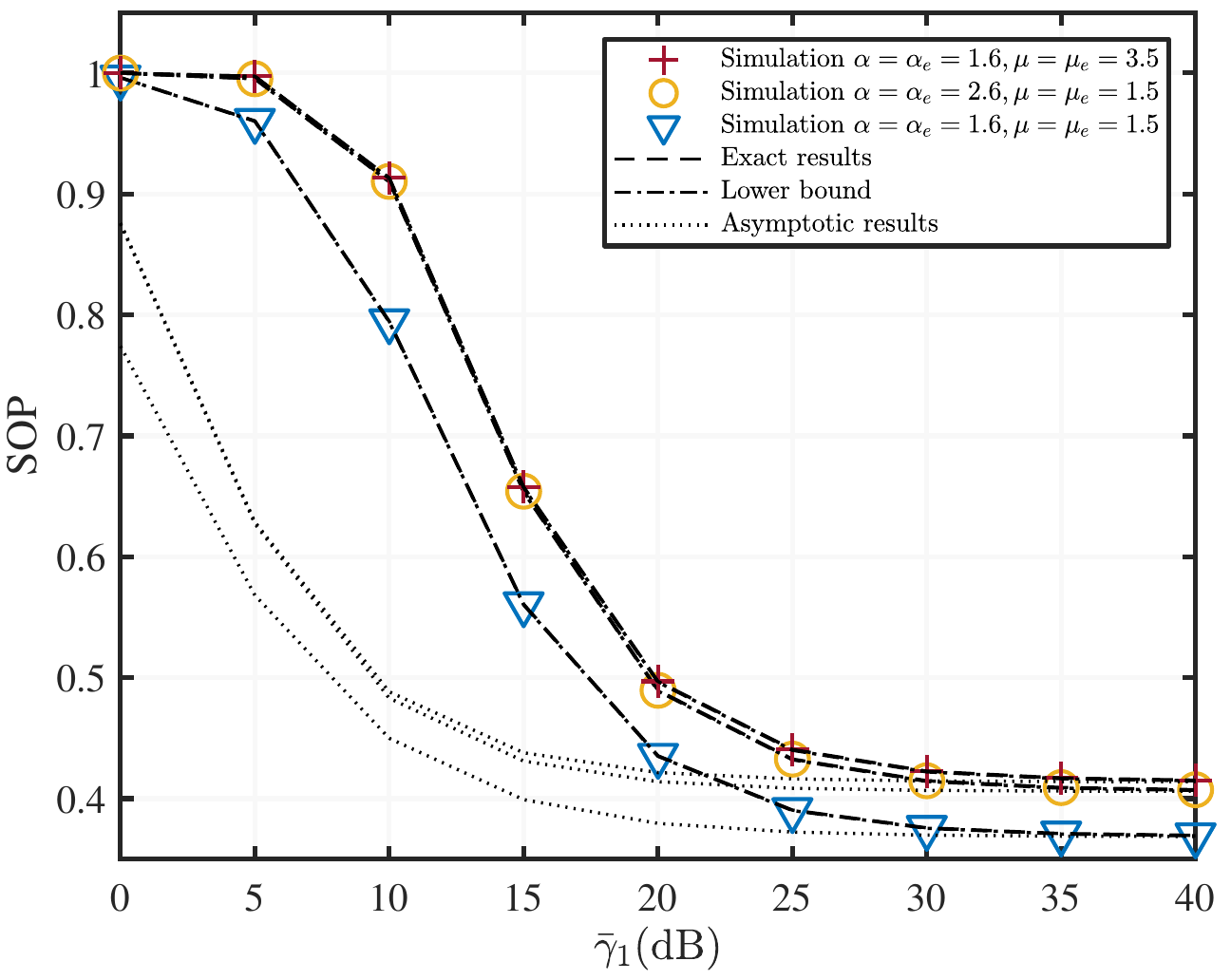}
\caption{SOP versus $\bar{\gamma}_1$ with various fading parameters $R_s=0.01$, $\bar{\gamma}_2=\bar{\gamma}_e=10$ dB, and UWOC channel parameter is [2.4,0.05]}
\label{fig.4}
\end{figure}

\begin{figure}[!htp]
\centering
\includegraphics[width=3.45in]{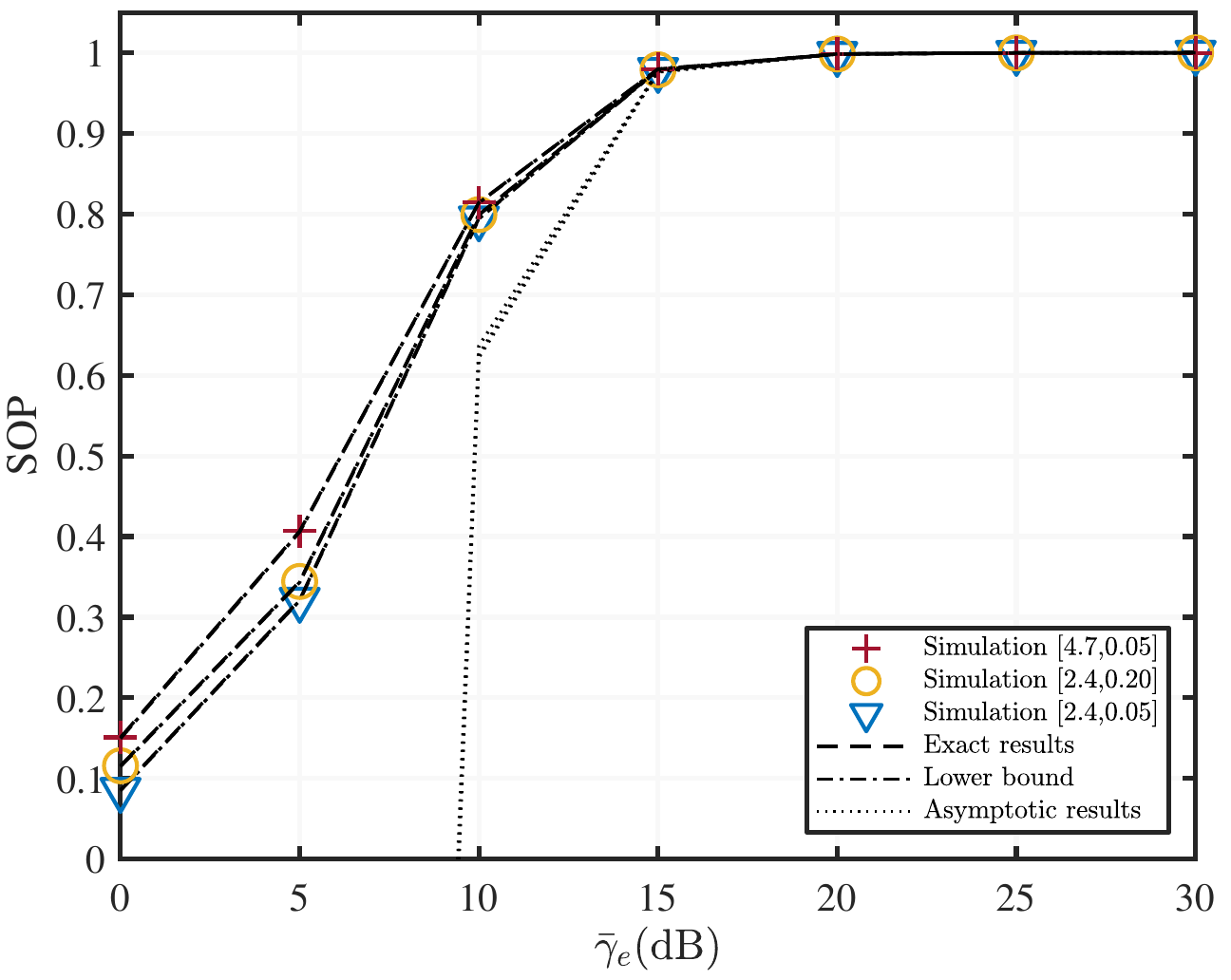}
\caption{SOP versus $\bar{\gamma}_e$ with various fading parameters when $\alpha=\alpha_e=1.6$, $\mu=\mu_e=1.5$, $R_s=0.01$, and $\bar{\gamma}_1=\bar{\gamma}_2=10$ dB}
\label{fig.5}
\end{figure}

\begin{figure}[!htp]
\centering
\includegraphics[width=3.45in]{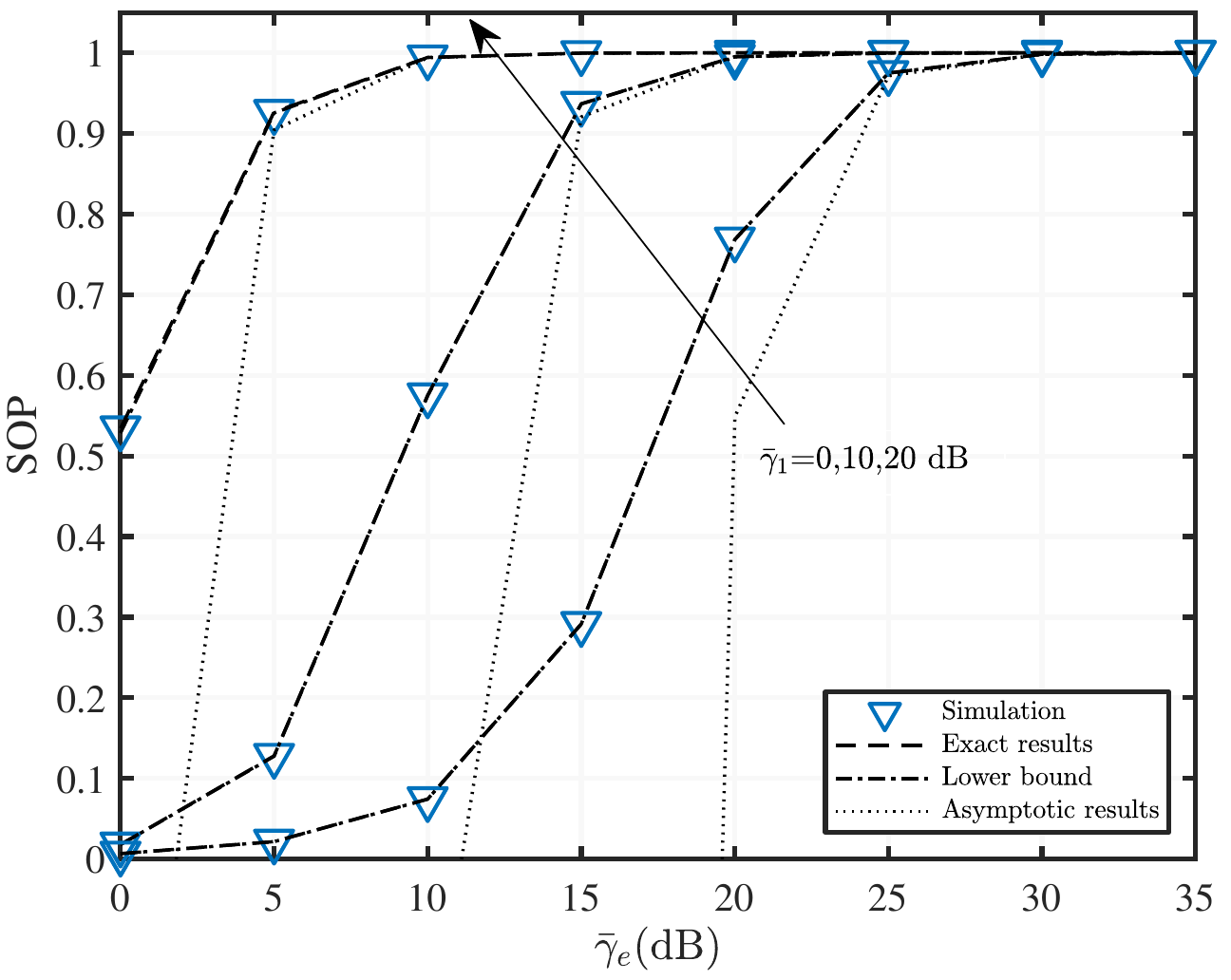}
\caption{SOP versus $\bar{\gamma}_e$ with various fading parameters when $\alpha=\alpha_e=1.6$, $\mu=\mu_e=1.5$, $R_s=0.01$, and $\bar{\gamma}_2=20$ dB}
\label{fig.6}
\end{figure}

Figure \ref{fig.2} shows the lower bound and the asymptotic SOP with average SNR of the SR link $\gamma_1$ for a mixed two-hop RF/UWOC system under different quality scenarios of UWOC channel. Both RF SR and SE links follow the $\alpha$-$\mu$ distribution and have the same parameters, where $\alpha=\alpha_e=1.6$, $\mu=\mu_e=1.5$. The average SNR of the SE and RD links are both set as $\bar{\gamma}_2=\bar{\gamma}_e=10$ dB. As shown in Fig. \ref{fig.2}, the exact theoretical results are almost identical to the simulation results, and both closely agree with the derived lower bound. Asymptotic results are tight when the average SNR is greater than 30 dB. Further, when the average SNR increases from 0 to 30 dB, SOP rapidly decreases. Also, SOP tends to saturate when the average SNR is between 30 and 40 dB. Given the cost of the relay battery replacement and engineering difficulties, the communication system should guarantee the SOP while cutting down on energy consumption. In practice, one should therefore select the optimal transmission power corresponding to the saturation starting point.

Figure \ref{fig.3} depicts the SOP variation versus the SR average SNR $\gamma_1$ for the mixed two-hop RF/UWOC system under three different eavesdropper interference levels, i.e., $\bar{\gamma}_e=3,0,-3$ dB. Parameters in Fig. \ref{fig.3} are set as follows: $\alpha=\alpha_e=1.6$, $\mu=\mu_e=1.5$, UWOC channel parameter is [2.4,0.05], and $\bar{\gamma}_2=0$ dB. It can be observed that the lower bounds closely match the exact results in the whole SNR region. The asymptotic result curve gradually coincides with the exact result curve when $\bar{\gamma}_1$ takes higher values starting from 20 dB. We can also observe that the SOP is monotonically increasing with $\bar{\gamma}_1$, assuming that the SNR of the SE link is a fixed value. Holding the other parameters constant, the larger the $\bar{\gamma}_e$, the lower the system SOP. In short, as the quality of the eavesdropping channel improves, the SOP performance of the system deteriorates.

Figure \ref{fig.4} indicates the effect of the variation in average SNR of the SR link on the SOP metric of a two-hop mixed RF/UWOC, with three different RF channel qualities. Evidently, SOP monotonically decreases with the increase of $\bar{\gamma}_1$, and SOP tends to saturate when $\bar{\gamma}_1 \geq$ 30 dB. Besides, Fig. \ref{fig.4} depicts that as the $\alpha-\mu$ value increases, the two-hop mixed RF/UWOC system secrecy performance worsens, and vice versa. This is because of the phenomena of severe nonlinearity and sparse clustering when the signals are propagating in a high $\alpha-\mu$ value RF channel, and poor RF channel quality makes it easier for eavesdroppers to intercept signals. As shown in Fig. \ref{fig.5}, as the $\bar{\gamma}_e$ progressively increases, the SOP value increases, the information intercepted by the eavesdropper increases, and the system's secrecy performance gradually decreases. Moreover, the asymptotic result is more accurate at $\bar{\gamma}_e$ greater than 15 dB.

In Fig. \ref{fig.6}, we set the same channel parameters as in Fig. \ref{fig.3}, except for setting the UWOC average SNR, i.e., $\bar{\gamma}_2=20$ dB. Fig. \ref{fig.6} shows that SOP increases with $\bar{\gamma}_e$ when the other parameters remain unchanged. The same interpretation of Fig. \ref{fig.5} can also be applied to Fig. \ref{fig.6}. Additionally, the rate at which the asymptotic results approach exact results varies for different SR average SNR. For $\bar{\gamma}_1=20$ dB, the asymptotic results begin to match the exact result starting at $\bar{\gamma}_e=5$ dB. Besides, the close match of the lower bound and the exact results demonstrates the robustness and accuracy of \eqref{soplb}.

In Fig. \ref{fig.7} -- Fig. \ref{fig.10}, We investigate the combined effect of the channel quality of both RF and UWOC channels on the PNZ metric of the  two-hop mixed RF/UWOC communication system.

\begin{figure}[!htp]
\centering
\includegraphics[width=3.45in]{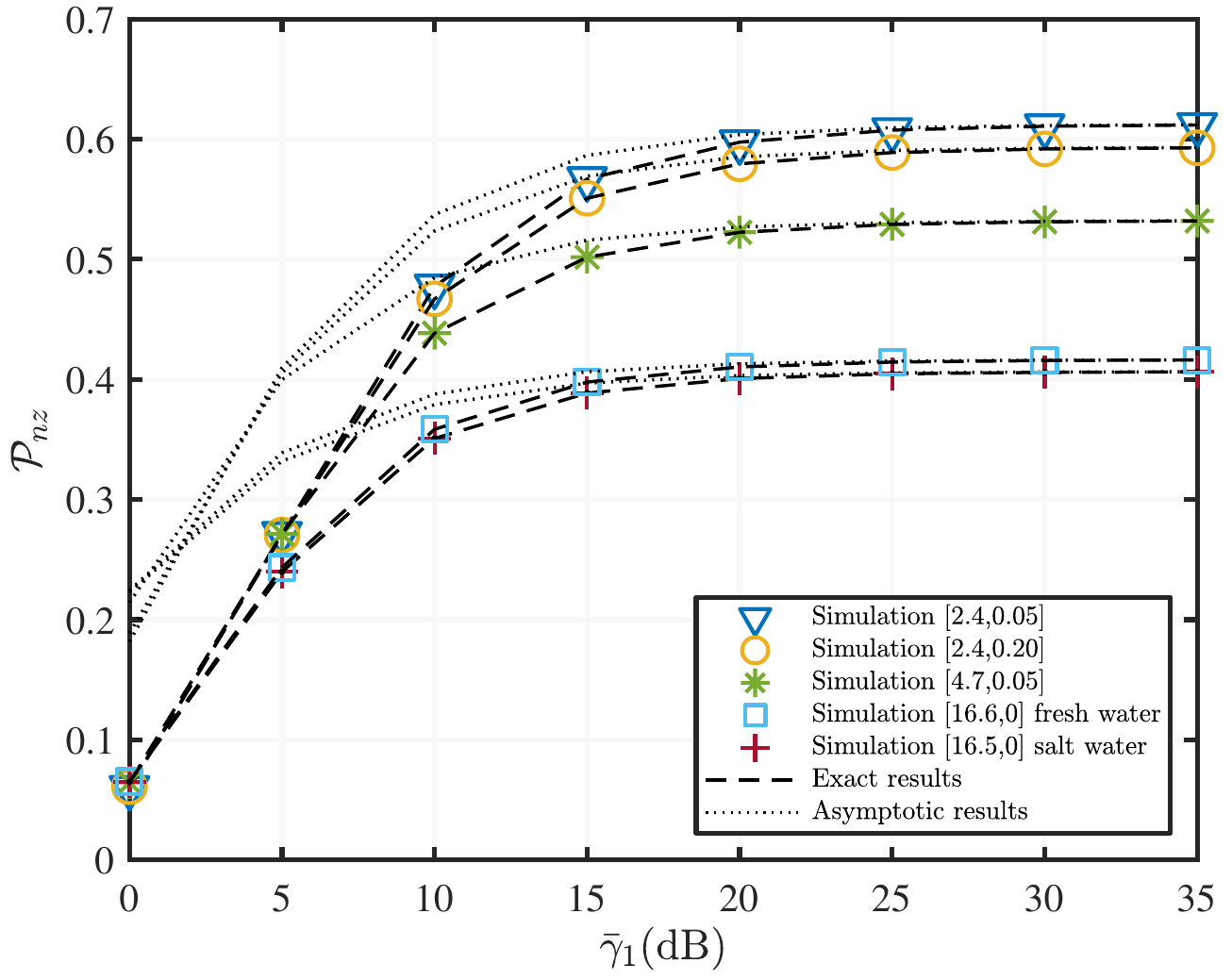}
\caption{$\mathcal{P}_{n z}$ versus $\bar{\gamma}_1$ with various fading parameters when $\alpha=\alpha_e=2.1$, $\mu=\mu_e=1.4$ and $\bar{\gamma}_e=\bar{\gamma}_2=0$ dB}
\label{fig.7}
\end{figure}

\begin{figure}[!htp]
\centering
\includegraphics[width=3.45in]{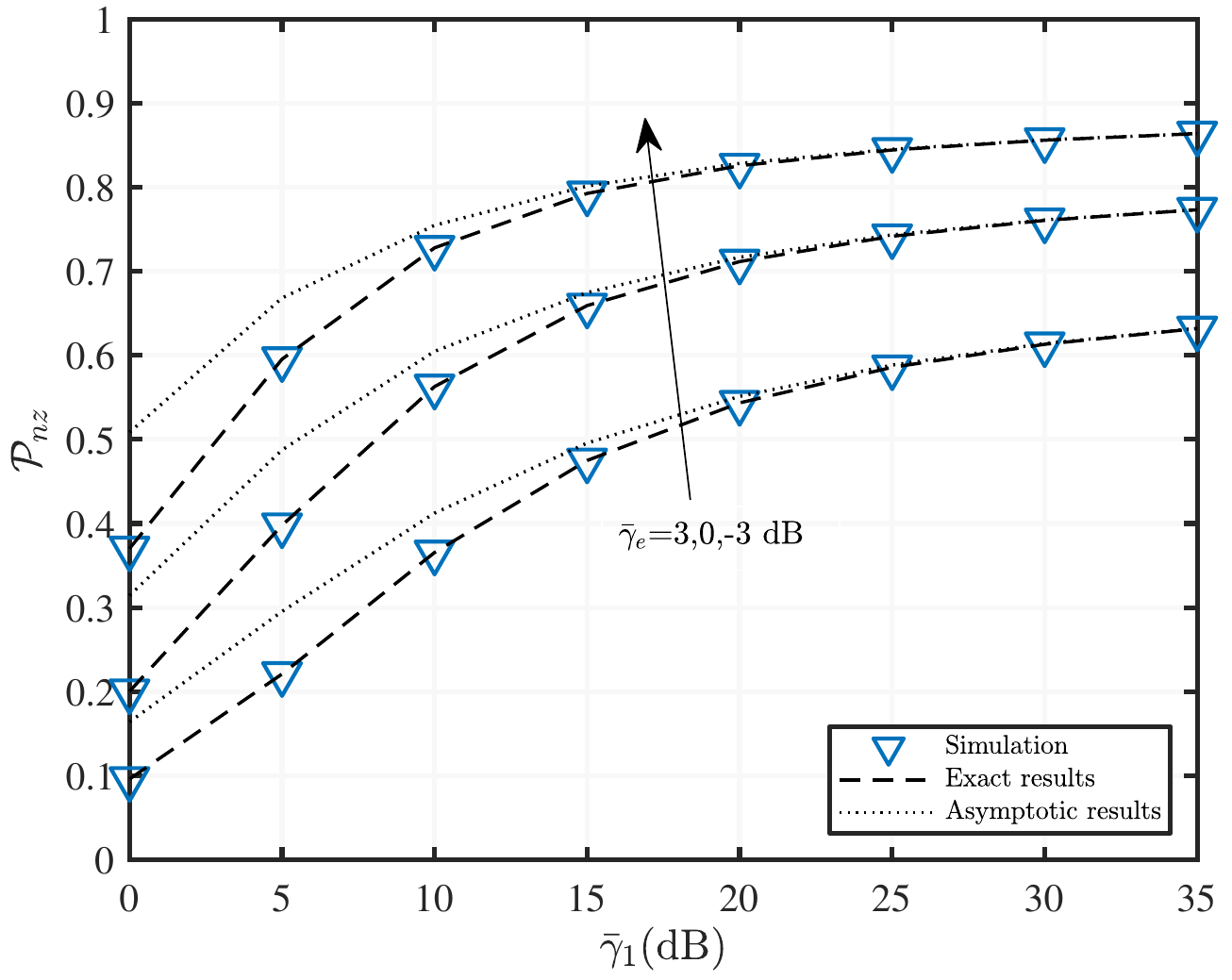}
\caption{$\mathcal{P}_{n z}$ versus $\bar{\gamma}_1$ with various fading parameters when $\alpha=\alpha_e=1.5$, $\mu=\mu_e=0.8$ and $\bar{\gamma}_2=0$ dB}
\label{fig.8}
\end{figure}

\begin{figure}[!htp]
\centering
\includegraphics[width=3.45in]{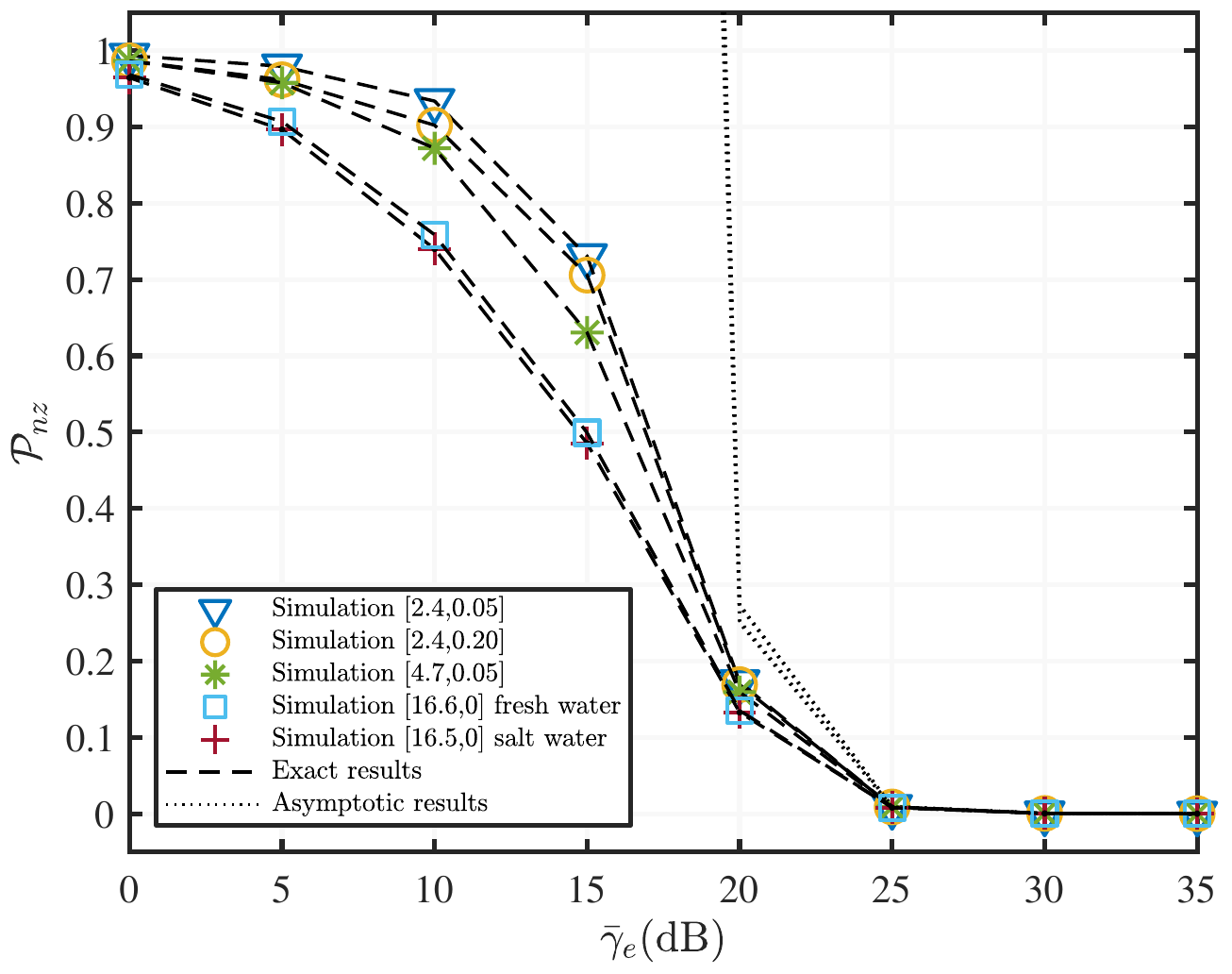}
\caption{$\mathcal{P}_{n z}$ versus $\bar{\gamma}_e$ with various fading parameters when $\alpha=\alpha_e=2.1$, $\mu=\mu_e=1.4$ and $\bar{\gamma}_1=\bar{\gamma}_2=20$ dB}
\label{fig.9}
\end{figure}

\begin{figure}[!htp]
\centering
\includegraphics[width=3.45in]{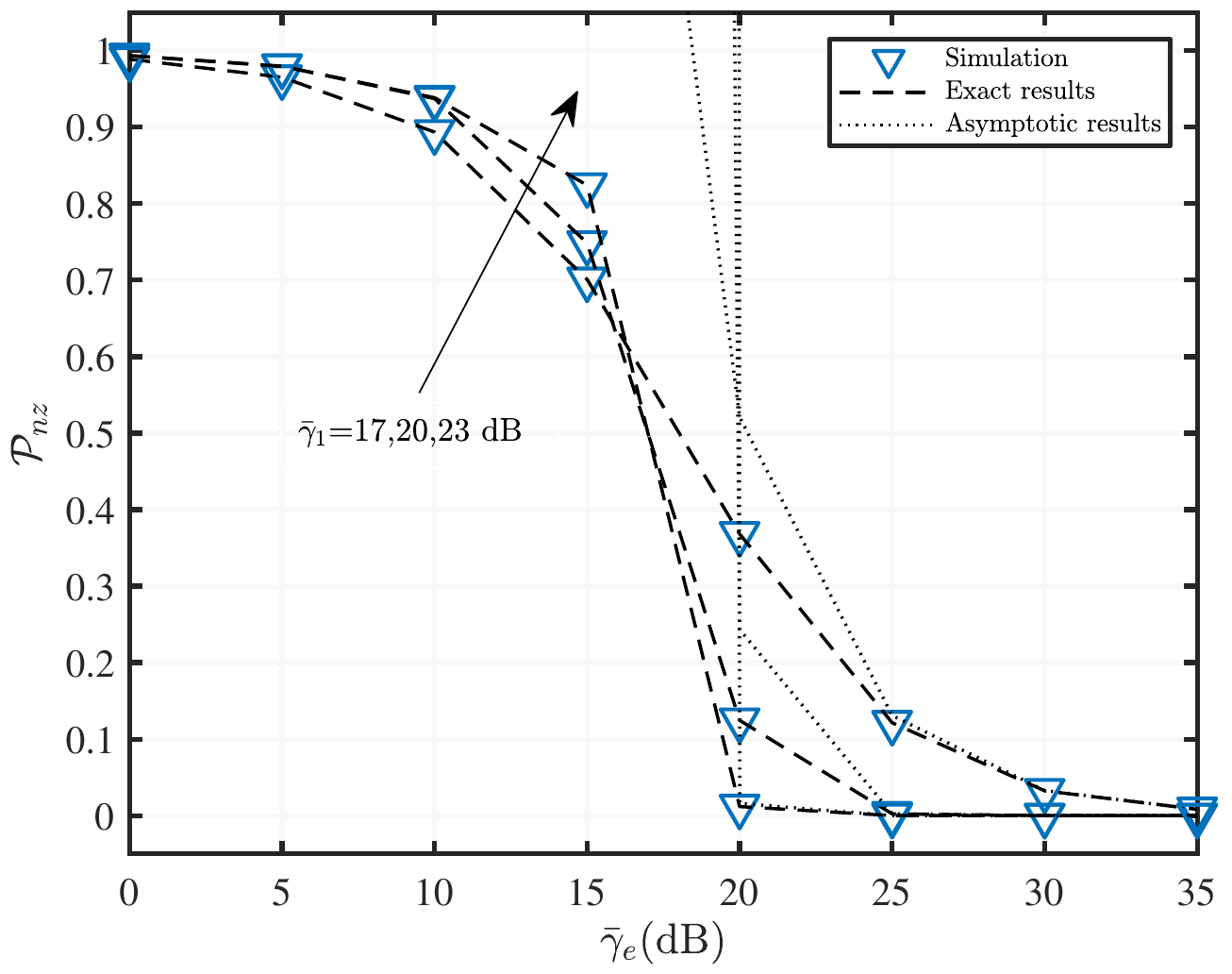}
\caption{$\mathcal{P}_{n z}$ versus $\bar{\gamma}_e$ with various fading parameters when $\alpha=\alpha_e=1.5$, $\mu=\mu_e=0.8$ and $\bar{\gamma}_2=20$ dB}
\label{fig.10}
\end{figure}

Figure \ref{fig.7} shows the effect of the SR link average SNR $\bar{\gamma}_1$ on the PNZ of the mixed RF/UWOC for different UWOC channel parameters. PNZ increases incrementally as $\bar{\gamma}_1$ increases, which indicates an increase in secrecy performance. It can be observed that PNZ decreases as the degree of turbulence increases, i.e., the higher the level of air bubbles and the larger the temperature gradient, the worse the secrecy performance in the system. Additionally, we depict the effects of salinity on UWOC performance in Fig. \ref{fig.7}. The salinity affects the system secrecy performance to a much lesser extent than the level of air bubble and temperature gradient.
This is because the generation and break-up of the air bubbles in the UWOC channels causes dramatic and random fluctuations of the underwater optical signals, which can significantly deteriorate the secrecy performance of the system. Fig. \ref{fig.7} shows that eavesdroppers may benefit from a low UWOC channel quality. On the contrary, in a high quality UWOC channel, the likelihood of an eavesdropper successfully eavesdropping is greatly reduced. Therefore, in practical applications, increasing the channel quality can increase the system transmission capacity and thus improve the system secrecy performance. Fig. \ref{fig.7} also shows that asymptotic results can quickly approach the exact result for poorer channels. For example, for a UWOC channel with channel parameters of [16.5,0], the asymptotic result can achieve a match with the exact value at  $\bar{\gamma}_1\geq 20$ dB. When the channel parameter set is [2.4, 0.05], the asymptotic result can only be accurate at  $\bar{\gamma}_1 >25$ dB. The remaining parameters are set as follows, $\bar{\gamma}_e=\bar{\gamma}_2=0$ dB, $\alpha=\alpha_e=2.1$, $\mu=\mu_e=1.4$.

In Fig. \ref{fig.8}, the RF channel parameters are $\alpha=\alpha_e=1.5$, $\mu=\mu_e=0.8$, and the UWOC channel parameters are [2.4,0.05]. We can explain the curves in Fig. \ref{fig.8} using a principle similar to Fig. \ref{fig.7}. Furthermore, in the case where $\bar{\gamma}_1$ remains unchanged, the smaller the $\bar{\gamma}_e$, the worse is the quality of the eavesdropping channel; therefore, leading to a PNZ performance degradation.

In addition to Fig. \ref{fig.8}, we analyzed the effect of the average SNR $\bar{\gamma}_1$ on the PNZ, as shown by Fig. \ref{fig.9} and Fig.\ref{fig.10}. The difference is that in Fig. \ref{fig.9} $\alpha=\alpha_e=2.1$, $\mu=\mu_e=1.4$ and $\bar{\gamma}_1=\bar{\gamma}_2=20$ dB. whereas the RF channel parameters in Fig. \ref{fig.10} are $\alpha=\alpha_e=1.5$ and $\mu=\mu_e=0.8$. It can be inferred from Fig. \ref{fig.9} and Fig. \ref{fig.10} that the asymptotic result only matches the exact value when $\bar{\gamma}_e$ is large, and the PNZ gradually decreases until it reaches zero.

\section{Conclusion}
We investigated the secrecy performance of a two-hop mixed RF/UWOC communication system using fixed-gain AF relaying. To allow the results to be more generic and applicable to more realistic physical scenarios, we model RF channels using the $\alpha$-$\mu$ distribution, which considers both the nonlinear of the transmission medium and multipath cluster characteristics, and model UWOC channels using the laboratory EGG distribution, which can account for different levels of air bubbles, temperature gradients, and salinity. Closed-form expressions for the PDF and the CDF of the two-hop end-to-end SNR were both derived in terms of the bivariate $H$-function. Based on these results, we  obtained a tight closed-form expression of the lower bound of the SOP and the exact closed-form expression of the PNZ. Furthermore, we also derived asymptotic expressions in simple functions for both SOP and PNZ to allow rapid numerical evaluation. Moreover, based on the asymptotic results, we presented an approach to determine the optimal transmitting power to maximize the energy efficiency, for given target performance of both SOP and PNZ. We fully investigated the effects of various existing phenomena of both RF and UWOC channels on the secrecy performance of the mixed RF/UWOC communication system. Also, our generalized theoretical framework is also applicable to various classical RF and underwater optical channel models including Rayleigh and Nakagami for RF channels and EG and Generalized Gamma for UWOC channels. Our results can be used in practical mixed security RF/UWOC communication systems design. The interesting topics for future work include: (i) to investigate the secrecy performance of a mixed RF/UWOC communication system using a energy-harvesting enabled relay with the aim of improving the system life time; (ii) to investigate the secrecy performance of a mixed RF/UWOC communication system using multiple relays with appropriate relaying selection algorithms.

\appendices
\renewcommand{\theequation}{\thesection.\arabic{equation}} 

\section{Proof of THEOREM 1}
\setcounter{equation}{0}

Using \eqref{gammaeq1}, we write the CDF of the  end-to-end SNR in the following form
\begin{IEEEeqnarray*}{rcl}
	 F_{\gamma_{eq}}(\gamma_{eq}) &=&\int_{0}^{\infty} \operatorname{Pr}\left[\frac{\gamma_{1} \gamma_{2}}{\gamma_{2}+C} \leq \gamma \mid \gamma_{2}\right] f_{\gamma_{2}}\left(\gamma_{2}\right) d \gamma_{2} \\ 
	 &=&1\!\!-\!\!\int_{\gamma}^{\infty} \bar{F}_{\gamma_{2}}\left(\frac{C \gamma}{x-\gamma}\right) f_{\gamma_{1}}(x) d x.\IEEEyesnumber \label{appcdfgammaeq}
\end{IEEEeqnarray*}

Substituting \eqref{ampdf} and \eqref{eggccdf} into \eqref{appcdfgammaeq} and replacing the integral variable $x$ with $z=x+\gamma$, after some simplifications, we can express \eqref{appcdfgammaeq} as
\begin{IEEEeqnarray*}{rcl}
	F_{\gamma _{eq}}\left(\gamma _{eq}\right)=1+I_1+I_2\IEEEyesnumber \label{appcdfgammaeq1}
\end{IEEEeqnarray*}
where
\begin{IEEEeqnarray*}{rcl}
	I_1&=&-\frac{\kappa  (1-\omega )}{\Gamma (a)}\int _0^{\infty }\!\!H_{0,1}^{1,0}\left[(z+\gamma ) \Lambda\left|\!\!\!\begin{array}{c}  \\ \left(-\frac{1}{\alpha }+\mu ,\frac{1}{\alpha }\right) \\\end{array}\right.\!\!\!\!\right] \\
	&&\times H_{1,2}^{2,0}\!\!\left[\frac{b^{-r} C \gamma }{z \mu _r}\middle|\!\!\!\begin{array}{c} (1,1) \\ (0,1),(a,\frac{r}{c}) \\\end{array}\!\!\!\right]dz\IEEEyesnumber\label{i11}
\end{IEEEeqnarray*}
and
\begin{IEEEeqnarray*}{rcl}
	I_2&=&-\kappa r \omega  \int _0^{\infty }H_{0,1}^{1,0}\!\!\left[\frac{C \gamma  \lambda ^{-r}}{z \mu _r}\middle|\!\!\!\begin{array}{c}  \\ (0,r) \\\end{array}\!\!\!\right] \\
	&&\times H_{0,1}^{1,0}\left[(z+\gamma ) \Lambda \left|\begin{array}{c}  \\ \left(-\frac{1}{\alpha }+\mu ,\frac{1}{\alpha }\right) \\\end{array}\right.\right]dz.\IEEEyesnumber\label{i21}
\end{IEEEeqnarray*}

To solve \eqref{i11}, we convert all the $H$-functions in \eqref{i11} into a line integral, and place the integral with respect to $x$ in the innermost part by rearranging the order of multiple integrals. Then, we have
\begin{IEEEeqnarray*}{rcl}
	I_1&=&\frac{\kappa  (1-\omega )}{4 \pi ^2 \Gamma (a)}\int _{\mathcal{L}}^t\frac{\Gamma (t)}{\Gamma (t+1)}\Gamma \left(a+\frac{r t}{c}\right)\left(\frac{b^r\mu _r}{C \gamma }\right)^t\\
	&&\times\!\!\! \int _{\mathcal{L}}^s\!\!\!\Lambda ^{-s} \Gamma \!\!\left(\frac{s}{\alpha }+\mu -\frac{1}{\alpha }\right) \!\!\!\int _0^{\infty }\!\!\!\!\!z^t(z+\gamma )^{-s}\!dzdsdt.\IEEEyesnumber
\end{IEEEeqnarray*}

By utilizing \cite[Eq. (3.197/1)]{i.s.gradshteynTableIntegralsSeries2007} to solve the integration of $z$, after some simplifications and using the definition of the bivariate $H$-function \cite[Eq. (2.57)]{mathaiHFunctionTheoryApplications2010}, we can finally express $I_1$ in \eqref{i11} in the following form
\begin{IEEEeqnarray*}{rcl}
	&&I_1= -\frac{\gamma  \kappa  (1-\omega )  }{\Gamma (a)}\\
	&&\times H_{1,0:0,2;1,1}^{0,1:2,0;0,1}\!\!\left[\!\!\!\begin{array}{c} \frac{b^{-r} C}{\mu _r} \\ \frac{1}{\gamma  \Lambda } \\\end{array}\middle|\!\!\!\begin{array}{ccccc} (2,1,1) &\!\!\!\!\!\! :\!\!\!\!\!\! &   &\!\!\!\!\!\! ;\!\!\!\!\!\! & \left(1+\frac{1}{\alpha }-\mu ,\frac{1}{\alpha }\right) \\
	   &\!\! \!\!\!\!: \!\!\!\!\!\!& (0,1),(a,\frac{r}{c}) &\!\!\!\!\!\! ;\!\!\!\!\!\! & (1,1) \\\end{array}\!\!\!\!\right].\\
	   \IEEEyesnumber \label{i12}
\end{IEEEeqnarray*}

We can solve \eqref{i21} in a similar way as we have solved \eqref{i11}. All $H$-functions are converted to the form of the line integrals and by rearranging the multiple integrals, the integral regarding $z$ is placed in the innermost part of the expression. Then, we have
\begin{IEEEeqnarray*}{rcl}
	I_2&=&\frac{\kappa  r \omega }{4 \pi ^2}\int _{\mathcal{L}}^s\Gamma (r s) \left(\frac{\lambda ^r \mu _r}{C \gamma }\right)^s \int _{\mathcal{L}}^t\Lambda^{-t} \Gamma \left(\frac{t}{\alpha }+\mu -\frac{1}{\alpha }\right) \\
	&&\times\int _0^{\infty }z^s (z+\gamma )^{-t}dzdtds.\IEEEyesnumber
\end{IEEEeqnarray*}

Again, we use \cite[Eq. (3.197/1)]{i.s.gradshteynTableIntegralsSeries2007} to solve the integration regarding $z$. Then use \cite[Eq. (2.57)]{mathaiHFunctionTheoryApplications2010} and some simplification, we obtain the following expression 
\begin{IEEEeqnarray*}{rcl}
	&&I_2=-\gamma  \kappa  r \omega  \\
	&&\times H_{1,0:0,2;1,1}^{0,1:2,0;0,1}\!\!\left[\!\!\!\begin{array}{c} \frac{C \lambda ^{-r}}{\mu _r} \\ \frac{1}{\gamma  \Lambda } \\\end{array}\middle|\!\!\!\begin{array}{ccccc} (2,1,1) &\!\!\!\!\!\! :\!\!\!\!\!\! &   \!\!\!&\!\!\! ; &\!\!\! \!\!\!\left(1+\frac{1}{\alpha }-\mu ,\frac{1}{\alpha }\right) \\
	   &\!\!\!\!\!\! : \!\!\!\!\!\!& (1,1),(0,r)\!\!\! &\!\!\! ; &\!\!\!\!\!\! (1,1) \\\end{array}\!\!\!\right].\\
	   \IEEEyesnumber \label{i22}
\end{IEEEeqnarray*}

Substituting \eqref{i12} and \eqref{i22} into \eqref{appcdfgammaeq1}, we obtian the exact closed-form expression for the CDF as shown by \eqref{gammaeqcdf}.

\section{Proof of THEOREM 2}
\setcounter{equation}{0}
The PDF of the end-to-end SNR can be obtained by using
\begin{IEEEeqnarray*}{rcl}
	 f(\gamma_{eq})=\frac{d F(\gamma_{eq})}{d \gamma_{eq}}.\IEEEyesnumber \label{apppdfgammaeq}
\end{IEEEeqnarray*}

Substituting \eqref{gammaeqcdf} into \eqref{apppdfgammaeq}, after some simplifications, we have
\begin{IEEEeqnarray*}{rcl}
	 f_{\gamma_{eq}}(\gamma_{eq}) &=& \frac{\mathrm{d} J_1}{\mathrm{d} \gamma_{eq}}+\frac{\mathrm{d} J_2}{\mathrm{d} \gamma_{eq}}\IEEEyesnumber\label{appgammaeqpdf} 
\end{IEEEeqnarray*}
where
\begin{IEEEeqnarray*}{rcl}
	J_1&=&\frac{\gamma_{eq}  \kappa  (1-\omega )}{4 \pi ^2 \Gamma (a)}\int _{\mathcal{L}}^t\int _{\mathcal{L}}^s\frac{1}{\Gamma (t)}\left(\frac{1}{\gamma_{eq}  \Lambda}\right)^t \Gamma (-s) \Gamma \left(a-\frac{r s}{c}\right) \\
	&&\times\Gamma (s+t-1) \Gamma \left(\frac{t}{\alpha }+\mu -\frac{1}{\alpha }\right) \left(\frac{b^{-r}C}{\mu _r}\right)^s\!dsdt\IEEEyesnumber
\end{IEEEeqnarray*}
and
\begin{IEEEeqnarray*}{rcl}
	J_2&=&\frac{\gamma_{eq}  \kappa  r \omega }{4 \pi ^2}\!\!\int _{\mathcal{L}}^t\!\int _{\mathcal{L}}^s\!\frac{1}{\Gamma (t)}\left(\frac{1}{\gamma_{eq}  \Lambda }\right)^t\!\!\!\!\Gamma (1\!-\!s) \Gamma (-r s)\Gamma (s+t-1)  \\
	&&\times\Gamma \left(\frac{t}{\alpha }+\mu -\frac{1}{\alpha }\right) \left(\frac{C \lambda ^{-r}}{\mu _r}\right)^sdsdt.\IEEEyesnumber
\end{IEEEeqnarray*}

By enabling the differential operation in \eqref{appgammaeqpdf}, after some rearrangements, we can represent the first and the second terms on the right of the equation \eqref{appgammaeqpdf} as
\begin{IEEEeqnarray*}{rcl}
	\frac{\mathrm{d} J_1}{\mathrm{d} \gamma_{eq}}&=&\frac{\kappa(1-\omega)}{4\pi^2\Gamma(a)}\int _{\mathcal{L}}^t\int _{\mathcal{L}}^s\!\!\!\frac{  (1-t)}{ \Gamma (t)}  C^s \Gamma (-s) (\gamma_{eq}\Lambda) ^{-t}   b^{-r s}\mu _r^{-s} \\
	&&\times\Gamma (s+t-1) \Gamma\!\! \left(\frac{t}{\alpha }+\mu -\frac{1}{\alpha }\right) \!\!\Gamma\!\! \left(a-\frac{r s}{c}\right)\!dsdt\IEEEyesnumber\label{appd1}
\end{IEEEeqnarray*}
and
\begin{IEEEeqnarray*}{rcl}
	\frac{\mathrm{d} J_2}{\mathrm{d} \gamma_{eq}}&=&\frac{\kappa r\omega}{4\pi^{2}}\int _{\mathcal{L}}^t\int _{\mathcal{L}}^s\frac{1}{ \Gamma (t)} (1-t)   C^s \Gamma (1-s) (\gamma_{eq} \Lambda)^{-t}  \lambda^{-r s} \\
	&&\times\mu _r^{-s} \Gamma (-r s) \Gamma (s+t-1) \Gamma \!\!\left(\frac{t}{\alpha }\!+\mu \!-\frac{1}{\alpha }\right)\!dsdt,\IEEEyesnumber \label{appd2}
\end{IEEEeqnarray*}
respectively.

After substitute \eqref{appd1} and \eqref{appd2} to \eqref{appgammaeqpdf}, and use the definition of bivariate $H$-function, we can derive the exact closed-form expression of PDF as shown in \eqref{gammaeqpdf}.

\section{Proof of THEOREM 3}
\setcounter{equation}{0}

Substituting \eqref{ampdf} and \eqref{gammaeqcdf} into \eqref{soplb1}, after some rearrangements, we have
\begin{IEEEeqnarray*}{rcl}
	\text{SOP}_L=1+Q_1+Q_2\IEEEyesnumber \label{appsoplb}
\end{IEEEeqnarray*}
where
\begin{IEEEeqnarray*}{rcl}
	&&Q_1=-\frac{  \Theta  \kappa  (1-\omega ) \kappa _e}{\Gamma (a)}\int _0^{\infty }\gamma H_{0,1}^{1,0}\!\!\left[\gamma  \Lambda _e\middle|\!\!\!\begin{array}{c}   \\ \left(-\frac{1}{\alpha _e}+\mu _e,\frac{1}{\alpha _e}\right) \\\end{array}\!\!\!\right]\\
	&&\times H_{1,0:0,2;1,1}^{0,1:2,0;0,1}\!\!\left[\!\!\!\begin{array}{c} \frac{b^{-r} C}{\mu _r} \\ \frac{1}{\gamma  \Theta  \Lambda } \\\end{array}\!\!\middle|\!\!\!\begin{array}{ccccc} (2,1,1) &\!\!\!\!\!\! : &\!\!\!\!\!   &\!\!\!\!\!\! ; \!\!\!\!\!\!& \left(1+\frac{1}{\alpha }-\mu ,\frac{1}{\alpha }\right) \\
	   &\!\!\!\!\!\! : & \!\!\!\!\!(0,1),(a,\frac{r}{c}) &\!\!\!\!\!\! ;\!\!\!\!\!\! & (1,1) \\\end{array}\!\!\!\!\right]\!\!d\gamma\\
	   \IEEEyesnumber\label{Q11}
\end{IEEEeqnarray*}
and
\begin{IEEEeqnarray*}{rcl}
	&&Q_2=-r \gamma  \Theta  \kappa  \omega  \kappa _e\int _0^{\infty }\gamma  H_{0,1}^{1,0}\!\!\left[\gamma  \Lambda _e\middle|\!\begin{array}{c}   \\ \left(-\frac{1}{\alpha _e}+\mu _e,\frac{1}{\alpha _e}\right) \\\end{array}\!\!\!\right] \\
	&&\times H_{1,0:0,2;1,1}^{0,1:2,0;0,1}\!\!\left[\!\!\!\begin{array}{c} \frac{C \lambda ^{-r}}{\mu _r} \\ \frac{1}{\gamma  \Theta  \Lambda } \\\end{array}\!\!\!\!\middle|\!\!\!\begin{array}{ccccc} (2,1,1)\!\!\! & \!\!\!: \!\!\!&\!\!   \!\!\!\!\!\!& ; \!\!& \!\!\!\left(1+\frac{1}{\alpha }-\mu ,\frac{1}{\alpha }\right) \\
	   \!\!\!&\!\!\! : \!\!\!& \!\!(1,1),\!(0,r) \!\!\!\!\!& ; \!\!& \!\!\!(1,1) \\\end{array}\!\!\!\!\right]\!\!d\gamma.\\
	   \IEEEyesnumber\label{Q21}
\end{IEEEeqnarray*}

To simplify \eqref{Q11} further, we first express the bivariate $H$-functions in \eqref{Q11} into the form of a double line integral, and then place the curve integral regarding $\gamma$ to the innermost level by rearranging \eqref{Q11}, we have
\begin{IEEEeqnarray*}{rcl}
	Q_1&=&\frac{\Theta  \kappa  (1-\omega ) \kappa _e}{4 \pi ^2 \Gamma (a)}\int _{\mathcal{L}}^t\frac{(\Theta  \Lambda )^{-t} }{\Gamma (t)}\Gamma \left(\frac{t}{\alpha}+\mu -\frac{1}{\alpha }\right)\\
	&&\times\int _{\mathcal{L}}^s\Gamma (-s) \Gamma \left(a-\frac{r s}{c}\right) \Gamma (s+t-1) \left(\frac{b^{-r}C}{\mu _r}\right)^s \\
	&&\times\int _0^{\infty }\!\!\!\!\gamma ^{1-t} H_{0,1}^{1,0}\!\!\left[\gamma  \Lambda _e\middle|\!\!\!\begin{array}{c}  \\ \left(-\frac{1}{\alpha _e}+\mu _e,\frac{1}{\alpha _e}\right) \\\end{array}\!\!\!\right]\!d\gamma dsdt.\IEEEyesnumber \label{Q12}
\end{IEEEeqnarray*}

Then, using \cite[Eq. 2.25.2/1]{vermaIntegralsInvolvingMeijer1966}, we can transform \eqref{Q12} into
\begin{IEEEeqnarray*}{rcl}
	Q_1&=&\frac{\Theta  \kappa  (1-\omega ) \kappa _e}{4 \pi ^2 \Gamma (a)}\int _{\mathcal{L}}^t\frac{(\Theta  \Lambda )^{-t} }{\Gamma (t)}\Gamma \left(\frac{t}{\alpha}+\mu -\frac{1}{\alpha }\right) \\
	&&\times\int _{\mathcal{L}}^s\Gamma (-s) \Gamma \left(a-\frac{r s}{c}\right) \Gamma (s+t-1) \\
	&&\times\Gamma \!\left(\frac{2-t}{\alpha_e}+\mu _e-\frac{1}{\alpha _e}\right) \Lambda _e^{t-2} \left(\frac{b^{-r} C}{\mu _r}\right)^s\!\!dsdt.\IEEEyesnumber  \label{Q13}
\end{IEEEeqnarray*}

Finally, converting the double curve integral into bivariate $H$-function using \cite[Eq. (2.57)]{mathaiHFunctionTheoryApplications2010}, after some simplifications, we obtain from \eqref{Q13} in an exact closed-form as
\begin{IEEEeqnarray*}{rcl}
	&&Q_1=-\frac{\Theta  \kappa  (1-\omega ) \kappa _e }{\Gamma (a) \Lambda _e^2}\\
	&&\times\! H_{1,0:1,2;0,2}^{0,1:1,1;2,0}\!\!\!\left[\!\!\!\!\begin{array}{c} \frac{\Lambda _e}{\Theta  \Lambda } \\ \frac{b^{-r} C}{\mu _r} \\\end{array}\!\!\!\middle|\!\!\!\begin{array}{ccccc} (2,\!1,\!1)\!\!\!\!\!\!& :\!\!\!\!\!\!\!& \left(1+\frac{1}{\alpha }-\mu ,\frac{1}{\alpha }\right)\!\! &\!\!\!\!\!\! ;\!\!\!\!\!\! &   \\
	   \!\!\!\!\!\!\!\!& :\!\!\!\!\!\!& \left(\frac{1+\alpha _e \mu _e}{\alpha _e},\!\frac{1}{\alpha _e}\right)\!,(1,\!1)& \!\!\!\!\!\!;\! \!\!\!\!\!& (0,1),(a,\frac{r}{c}) \\\end{array}\!\!\!\!\right].\\
	   \IEEEyesnumber \label{Q14}
\end{IEEEeqnarray*}

To process \eqref{Q21} further, we first convert the bivariate $H$-function in \eqref{Q21} into the form of one double curve integral using \cite[Eq. (2.55)]{mathaiHFunctionTheoryApplications2010}. After placing the line integral of $\gamma$ into the innermost layer, we can transform \eqref{Q21} into
\begin{IEEEeqnarray*}{rcl}
	Q_2&=&\frac{\Theta  \kappa  r \omega  \kappa _e}{4 \pi ^2}\int _{\mathcal{L}}^t\frac{(\Theta  \Lambda )^{-t} }{\Gamma (t)}\Gamma \left(\frac{t}{\alpha }+\mu -\frac{1}{\alpha}\right)\\
	&&\times\int _{\mathcal{L}}^s\Gamma (1-s) \Gamma (-r s) \Gamma (s+t-1) \left(\frac{C \lambda ^{-r}}{\mu _r}\right)^s\\
	&&\times\int _0^{\infty}\!\!\!\!\gamma ^{1-t} H_{0,1}^{1,0}\!\!\left[\gamma  \Lambda _e\middle|\!\!\!\!\begin{array}{c}  \\ \left(-\frac{1}{\alpha _e}+\mu _e,\frac{1}{\alpha _e}\right) \\\end{array}\!\!\!\!\right]\!\!d\gamma dsdt.\IEEEyesnumber \label{Q22}
\end{IEEEeqnarray*}

Subsequently, using \cite[Eq. 2.25.2/1]{vermaIntegralsInvolvingMeijer1966}, we express the innermost curve integral in \eqref{Q22} in the form of the product of Gamma functions. Then, we can write \eqref{Q22} as
\begin{IEEEeqnarray*}{rcl}
	Q_2&=&\frac{\Theta  \kappa  r \omega  \kappa _e}{4 \pi ^2}\int _{\mathcal{L}}^t\frac{(\Theta  \Lambda )^{-t} }{\Gamma (t)}\Gamma \left(\frac{t}{\alpha }+\mu -\frac{1}{\alpha}\right)\\
	&&\times\int _{\mathcal{L}}^s\Gamma (1-s) \Gamma (-r s) \Gamma (s+t-1) \\
	&&\times\Gamma \!\left(\frac{2-t}{\alpha _e}+\mu _e-\frac{1}{\alpha _e}\right)\Lambda _e^{t-2}\! \left(\frac{C \lambda ^{-r}}{\mu _r}\right)^s\!\!\!dsdt.\IEEEyesnumber \label{Q23}
\end{IEEEeqnarray*}

Subsequently, based on the same steps as for the derivation of \eqref{Q14}, eq. \eqref{Q23} can be expressed in exact closed-form as
\begin{IEEEeqnarray*}{rcl}
	&&Q_2=-\frac{\Theta  \kappa  r \omega  \kappa _e }{\Lambda _e^2}\\
	&&\times H_{1,0:1,2;0,2}^{0,1:1,1;2,0}\!\!\left[\!\!\!\!\begin{array}{c} \frac{\Lambda _e}{\Theta  \Lambda } \\ \frac{C \lambda ^{-r}}{\mu _r} \\\end{array}\!\!\!\!\middle|\!\!\!\!\begin{array}{ccccc} (2,\!1,\!1)\!\!\!\!\!\! & : &\!\!\!\!\! \!\!\left(1+\frac{1}{\alpha }-\mu ,\frac{1}{\alpha }\right)\!\! \!\!\!\!\!\!& ;\!\!\!\!\!\! & \!  \\
	  \!\!\!\!\!\!\!\!\!& : & \!\!\!\!\!\!\!\left(\frac{1+\alpha _e \mu _e}{\alpha _e},\frac{1}{\alpha _e}\right)\!,\!(1,1) \!\!\!\!\!\!& ;\!\! \!\!\!\!& (1,1),\!(0,r) \\\end{array}\!\!\!\!\right].\\
	   \IEEEyesnumber \label{Q24}
\end{IEEEeqnarray*}

After substituting \eqref{Q14} and \eqref{Q24} into \eqref{appsoplb}, we can finally obtain the closed-form expression of $\text{SOP}_L$ in \eqref{soplb}.

\section{Proof of Corollary 3.1}
\setcounter{equation}{0}
To derive the asymptotic expression of $\text{SOP}$, we need to derive the asymptotic expressions of the first and the second bivariate $H$-function on the right-hand side of \eqref{soplb}, which are denoted by $O_1$ and $O_2$, respectively. We consider two cases: (a) $\gamma_1 \to \infty$ and (b) $\gamma_e \to \infty$.

\subsection{Case \texorpdfstring{$\gamma_1 \to \infty$}{}}
For the case $\gamma_1 \to \infty$, we first focus on deriving asymptotic expression for $O_1$. Observe that as $\gamma_1$ tends to infinity, $\frac{\theta \Lambda}{\Lambda e}$ tends to zero. Thus, we first express the bivariate $H$-function in the form of one double curve integral, and express the curve integral containing $\frac{\theta \Lambda}{\Lambda e}$ in the form of an $H$- function. Then, we have  
\begin{IEEEeqnarray*}{rcl}
	O_1&=&\frac{i \Theta  \kappa  (1-\omega ) \kappa _e}{2 \pi  \Gamma (a) \Lambda _e^2}\int _{\mathcal{L}}^t  \Gamma (-t) \Gamma \left(a-\frac{r t}{c}\right)\left(\frac{b^{-r} C}{\mu _r}\right)^t \\
	&&\times H_{2,2}^{2,1}\!\!\left[\frac{\Theta  \Lambda }{\Lambda _e}\middle|\!\!\!\begin{array}{c} \left(1-\frac{1}{\alpha _e}-\mu _e,\frac{1}{\alpha _e}\right),(0,1) \\ (-1+t,1),(-\frac{1}{\alpha }+\mu ,\frac{1}{\alpha }) \\\end{array}\!\!\!\right]dt.\IEEEyesnumber \label{O11}
\end{IEEEeqnarray*}

It is easy to observe that the $H$-function in \eqref{O11} contains two poles: $(1-t)$ and $(1-\alpha\mu)$. According to \cite{alhennawiClosedFormExactAsymptotic2016}, when the argument tends to zero, the asymptotic value of the $H$-function can be expressed as the residue of the closest pole to the left of the integration path $l$. Therefore, by utilizing  \cite[Eq. (1.8.4)]{kilbasHtransformsTheoryApplications2004}, we can express \eqref{O11} as
\begin{IEEEeqnarray*}{rcl}
	O_1&=&\frac{i \kappa  (1-\omega ) \kappa _e}{2 \pi  \Lambda  \Gamma (a) \Lambda _e}\int _{\mathcal{L}}^t\frac{\Gamma (-t) }{\Gamma(1-t)}\Gamma \left(a-\frac{r t}{c}\right)\Gamma \left(\mu -\frac{t}{\alpha }\right) \\
	&&\times\Gamma \left(\frac{t}{\alpha _e}+\mu _e\right) \left(\frac{b^{-r} C \Theta  \Lambda }{\Lambda _e \mu _r}\right)^tdt.\IEEEyesnumber \label{O12}
\end{IEEEeqnarray*}

Following some simplifications, and using the definition of the $H$-function, we can transform \eqref{O12} into the following form
\begin{IEEEeqnarray*}{rcl}
	O_1\!\!=\!\!-\frac{\kappa  (1-\omega ) \kappa _e }{\Lambda  \Gamma (a) \Lambda _e}H_{3,2}^{1,3}\!\!\left[\!\frac{b^r \Lambda _e \mu _r}{C \Theta  \Lambda }\!\!\middle|\!\!\!\begin{array}{c} (1,\!1),\!(1-a,\frac{r}{c}),(1-\mu ,\frac{1}{\alpha }) \\ \!\!\left(\mu _e,\frac{1}{\alpha _e}\right),(0,1) \\\end{array}\!\!\!\!\right].\\
	\IEEEyesnumber \label{O13}
\end{IEEEeqnarray*}

Next, we derive the asymptotic expression for $O_2$. Observing that $O_2$ and $O_1$ have a similar structure, we can readily transform $O_2$ into the following form
\begin{IEEEeqnarray*}{rcl}
	O_2&=&-\frac{i \Theta  \kappa  r \omega  \kappa _e}{2 \pi  \Lambda _e^2}\int _{\mathcal{L}}^t  \Gamma (1-t) \Gamma (-r t) \left(\frac{C \lambda ^{-r}}{\mu_r}\right)^t\\ &&\times H_{2,2}^{2,1}\!\!\left[\frac{\Theta  \Lambda }{\Lambda _e}\middle|\!\!\!\begin{array}{c} \left(1-\frac{1}{\alpha _e}-\mu _e,\frac{1}{\alpha _e}\right),(0,1) \\ (-1+t,1),(-\frac{1}{\alpha }+\mu ,\frac{1}{\alpha }) \\\end{array}\!\!\!\right]dt.\\
	\IEEEyesnumber \label{O21}
\end{IEEEeqnarray*}

Similarly, we again use the residue of the pole $(1-t)$ to represent the asymptotic value of the $H$-function in \eqref{O21} as the argument tends to zero. Then, we have
\begin{IEEEeqnarray*}{rcl}
	O_2&=&\frac{i \kappa  r \omega  \kappa _e}{2 \pi  \Lambda  \Lambda _e}\int _{\mathcal{L}}^t\Gamma (-r t) \Gamma \left(\mu -\frac{t}{\alpha }\right) \Gamma\left(\frac{t}{\alpha _e}+\mu _e\right)\\
	&&\times \left(\frac{C \Theta  \lambda ^{-r} \Lambda }{\Lambda _e \mu _r}\right)^tdt.\IEEEyesnumber\label{O22}
\end{IEEEeqnarray*}

By using the definition of the $H$-function, we can transform \eqref{O22} into the following form
\begin{IEEEeqnarray*}{rcl}
	O_2=-\frac{\kappa  r \omega  \kappa _e }{\Lambda  \Lambda _e}H_{2,1}^{1,2}\!\!\left[\frac{\lambda ^r \Lambda _e \mu _r}{C \Theta  \Lambda }\middle|\!\!\!\begin{array}{c} (1,r),(1-\mu ,\frac{1}{\alpha }) \\ \left(\mu _e,\frac{1}{\alpha _e}\right) \\\end{array}\!\!\!\right].\IEEEyesnumber \label{O23}
\end{IEEEeqnarray*}

Substituting \eqref{O13} and \eqref{O23} into \eqref{soplb}, we obtain the asymptotic expression for SOP for the case $\gamma_1 \to \infty$ as shown in \eqref{sopa}.

\subsection{Case \texorpdfstring{$\gamma_e \to \infty$}{}}
Now, we focus on the case $\gamma_e \to \infty$. Obviously, as $\gamma_e$ tends to infinity, $\frac{\theta \Lambda}{\Lambda e}$ tends to infinity. Thus, using \cite[Eq. (1.5.9)]{kilbasHtransformsTheoryApplications2004} and a similar approach to that used in case $\gamma_1 \to \infty$, we can easily obtain closed-form expressions for $O_1$ and $O_2$ for case $\gamma_e \to \infty$, as
\begin{IEEEeqnarray*}{rcl}
	O_1&=&-\frac{(1-\omega ) \alpha _e}{\Gamma (a) \Gamma (\mu ) \Gamma \left(\mu _e\right) \Gamma \left(\alpha _e \mu _e+1\right)}\Gamma\left(\mu +\frac{\alpha _e \mu _e}{\alpha }\right) \\
	&&\times\left(\frac{\Lambda _e}{\Theta  \Lambda }\right)^{\alpha _e \mu _e} H_{2,1}^{1,2}\!\!\left[\frac{b^r\mu _r}{C}\middle|\!\!\!\begin{array}{c} (1,1),(1-a,\frac{r}{c}) \\ \left(\alpha _e \mu _e,1\right) \\\end{array}\!\!\!\right] \label{Oe13}
\end{IEEEeqnarray*}
and
\begin{IEEEeqnarray*}{rcl}
	O_2&=&-\frac{r \omega  \alpha _e}{\Gamma (\mu ) \Gamma \left(\mu _e\right) \Gamma \left(\alpha _e \mu _e+1\right)}\Gamma \left(\mu +\frac{\alpha_e \mu _e}{\alpha }\right) \\
	&&\times\left(\frac{\Lambda _e}{\Theta  \Lambda }\right)^{\alpha _e \mu _e} H_{2,1}^{1,2}\!\!\left[\frac{\lambda ^r \mu _r}{C}\middle|\!\!\!\begin{array}{c} (0,1),(1,r) \\ \left(\alpha _e \mu _e,1\right) \\\end{array}\!\!\!\right] \label{Oe23}
\end{IEEEeqnarray*}
respectively.

Substituting \eqref{Oe13} and \eqref{Oe23} into \eqref{soplb}, we obtain the asymptotic expression for SOP for the case $\gamma_e \to \infty$ as shown in \eqref{sopae}.

\section{Proof of THEOREM 4}
\setcounter{equation}{0}
Substituting \eqref{amcdfform2} and \eqref{gammaeqpdf} into \eqref{pnz1}, after some simplifications, we can transform the PNZ expression in \eqref{pnz1} to
\begin{IEEEeqnarray*}{rcl}
	\mathcal{P}_{n z}=T_1+T_2\IEEEyesnumber
\end{IEEEeqnarray*}
where
\begin{IEEEeqnarray*}{rcl}
	&&T_1=\int _0^{\infty }\frac{\kappa  (1-\omega ) \kappa _e}{\Gamma (a) \Lambda _e}H_{1,2}^{1,1}\!\!\left[\gamma  \Lambda _e\middle|\!\!\!\begin{array}{c} (1,1) \\ \left(\mu _e,\frac{1}{\alpha _e}\right),(0,1) \\\end{array}\!\!\!\right]\\
	&&\times H_{1,0:1,1;0,2}^{0,1:0,1;2,0}\!\!\left[\!\!\!\!\begin{array}{c} \frac{1}{\gamma  \Lambda } \\ \frac{b^{-r} C}{\mu _r} \\\end{array}\!\!\middle|\!\!\!\begin{array}{ccccc} (2,1,1)\!\!\!\!\!\! & : \!\!\!\!\!\!& \left(1+\frac{1}{\alpha }-\mu ,\frac{1}{\alpha }\right) \!\!\!\!\!& ; & \!\!\!\!\!\!  \\
	  \!\!\!\!\!\! & :\!\!\!\!\!\! & (2,1)\!\!\!\!\! & ; &\!\!\!\!\! (0,1),\!(a,\frac{r}{c}) \\\end{array}\!\!\!\!\right]\!\!d\gamma\\
	  \IEEEyesnumber\label{t11}
\end{IEEEeqnarray*}
and
\begin{IEEEeqnarray*}{rcl}
	&&T_2=\int _0^{\infty }\frac{r \kappa  \omega  \kappa _e}{\Lambda _e}H_{1,2}^{1,1}\!\!\left[\gamma  \Lambda _e\middle|\!\!\!\begin{array}{c} (1,1) \\ \left(\mu _e,\frac{1}{\alpha _e}\right),(0,1) \\\end{array}\!\!\!\right] \\
	&&\times H_{1,0:1,1;0,2}^{0,1:0,1;2,0}\!\!\left[\!\!\!\!\begin{array}{c} \frac{1}{\gamma  \Lambda } \\ \frac{C \lambda ^{-r}}{\mu _r} \\\end{array}\!\!\middle|\!\!\!\begin{array}{ccccc} (2,1,1)\!\!\!\!\!\! & :\!\!\!\!\!\! & \left(1+\frac{1}{\alpha }-\mu ,\frac{1}{\alpha }\right) \!\!\!\!\!\!& ; &\!\!\!\!\!   \\
	   \!\!\!\!\!\!& :\!\!\!\!\!\! & (2,1)\!\!\!\!\!\! & ; &\!\!\!\!\! (1,1),(0,r) \\\end{array}\!\!\!\!\right]\!\!d\gamma.\\
	   \IEEEyesnumber\label{t21}
\end{IEEEeqnarray*}

Representing the bivariate $H$-function into the form of one double line integral and moving the line integral regarding $\gamma$ to the innermost level, we can re-write \eqref{t11} as
\begin{IEEEeqnarray*}{rcl}
	T_1&=&-\frac{\kappa  (1-\omega ) \kappa _e}{4 \pi ^2 \Gamma (a) \Lambda _e}\int _{\mathcal{L}}^t\Gamma (-t) \Gamma \left(a-\frac{r t}{c}\right) \left(\frac{b^{-r}C}{\mu _r}\right)^{t}\\
	&&\times\int _{\mathcal{L}}^s\frac{\Lambda ^{-s}  }{\Gamma(s-1)}\Gamma (s+t-1) \Gamma \left(\frac{s}{\alpha }+\mu -\frac{1}{\alpha }\right)\\
	&&\times\int _0^{\infty }\!\!\gamma ^{-s} H_{1,2}^{1,1}\!\!\left[\gamma  \Lambda _e\middle|\!\!\!\begin{array}{c} (1,1) \\ \left(\mu _e,\frac{1}{\alpha _e}\right),(0,1) \\\end{array}\!\!\!\right]\!d\gamma dsdt.\IEEEyesnumber \label{t12}
\end{IEEEeqnarray*}

Afterwards, using the same technique as that used for deducing \eqref{Q14} and \eqref{Q24}, we can express \eqref{t12} as
\begin{IEEEeqnarray*}{rcl}
	&&T_1=\frac{\kappa  (1-\omega ) \kappa _e }{\Gamma (a) \Lambda _e^2}\\
	&&\times \!H_{1,0:0,2;1,2}^{0,1:2,0;1,1}\!\!\left[\!\!\!\!\begin{array}{c} \frac{b^{-r} C}{\mu _r} \\ \frac{\Lambda _e}{\Lambda } \\\end{array}\!\!\!\middle|\!\!\!\begin{array}{ccccc} (2,\!1,\!1) \!\!\!\!\!\!& : &\!\!\!\!\!\!   &\!\!\!\!\!\! ;\!\!\!\!\!\! & \left(1+\frac{1}{\alpha }-\mu ,\frac{1}{\alpha }\right) \!\!\!\\
	   \!\!\!\!\!\!\!\!& : &\!\!\!\!\! (0,\!1),\!(a,\!\frac{r}{c}) &\!\!\!\!\! ;\!\!\!\!\!\! & \left(\frac{1+\alpha _e \mu _e}{\alpha _e},\frac{1}{\alpha _e}\right)\!,\!(1,\!1) \\\end{array}\!\!\!\!\right].\\
	   \IEEEyesnumber
\end{IEEEeqnarray*}

Similarly, $T_2$ in \eqref{t21} can be transformed into
\begin{IEEEeqnarray*}{rcl}
	T_2&=&-\frac{\kappa  r \omega  \kappa _e}{4 \pi ^2 \Lambda _e}\int _{\mathcal{L}}^t\Gamma (1-t) \Gamma (-r t) \left(\frac{C \lambda ^{-r}}{\mu _r}\right)^t\\
	&&\times\int _{\mathcal{L}}^s\frac{\Lambda ^{-s} }{\Gamma (s-1)}\Gamma (s+t-1) \Gamma \left(\frac{s}{\alpha }+\mu -\frac{1}{\alpha }\right)\\
	&&\times\int _0^{\infty}\!\!\!\!\gamma ^{-s} H_{1,2}^{1,1}\!\!\left[\gamma  \Lambda _e\middle|\!\!\!\begin{array}{c} (1,1) \\ \left(\mu _e,\frac{1}{\alpha _e}\right),(0,1) \\\end{array}\!\!\!\right]\!d\gamma dsdt\IEEEyesnumber
\end{IEEEeqnarray*}
which is then expressed as
\begin{IEEEeqnarray*}{rcl}
	&&T_2=\frac{\kappa  r \omega  \kappa _e }{\Lambda _e^2}\\
	&&\times H_{1,0:0,2;1,2}^{0,1:2,0;1,1}\!\!\left[\!\!\!\!\begin{array}{c} \frac{C \lambda ^{-r}}{\mu _r} \\ \frac{\Lambda _e}{\Lambda } \\\end{array}\!\!\!\!\middle|\!\!\!\begin{array}{ccccc} (2,\!1,\!1)\!\!\!\!\! & : &\!\!\!\!  \!\!\!\!\!\!\!\! & ; \!\!\!\!\!& \left(1+\frac{1}{\alpha }-\mu ,\frac{1}{\alpha }\right)\!\!\! \\
	  \!\!\!\!\!\!\! & : &\!\!\!\!\! (1,\!1),\!(0,r) \!\!\!\!\!& ; \!\!\!\!\!& \left(\frac{1+\alpha _e \mu _e}{\alpha _e},\frac{1}{\alpha _e}\right)\!,\!(1,\!1) \\\end{array}\!\!\!\!\right]\\
	   \IEEEyesnumber
\end{IEEEeqnarray*}
using \cite[Eq. 2.25.2/1]{vermaIntegralsInvolvingMeijer1966}. 

\bibliographystyle{IEEEtran}
\bibliography{IEEEabrv,153.bib}
\end{document}